\documentclass[english,a4paper,11pt]{article}

\usepackage{amsmath}
\usepackage{amssymb}
\usepackage{exscale}
\usepackage{amsthm}
\usepackage{epsfig}
\usepackage{verbatim}
\usepackage{setspace}
\usepackage{psfrag}
\usepackage{longtable}
\usepackage{tabularx}       \usepackage{multicol}

\newtheorem{theorem}{Theorem}[section]
\newtheorem{proposition}[theorem]{Proposition}
\newtheorem{lemma}[theorem]{Lemma}
\newtheorem{corollary}[theorem]{Corollary}

\newtheorem{remark}[theorem]{Remark}

\newcommand{\setmid}{\,|\,}
\newcommand{\Setmid}{\,\Big|\,}
\newcommand{\N}{{\mathbb{N}}}
\newcommand{\Z}{{\mathbb{Z}}}

\newcommand{\R}{{\mathbb{R}}}

\newcommand{\A}{{A_\varepsilon}}

\newcommand{\Expec}{{\mathbb{E}}}

\newcommand{\Ap}{\mathcal{C}}
\newcommand{\EN}{\mathcal{N}}

\newcommand{\disc}{d^*_\infty}

\DeclareMathOperator{\globalbest}{\textsc{global}} \DeclareMathOperator{\current}{\textsc{current}} \newcommand{\thresh}{\texttt{Thresholds}}

\newcommand{\basic}{\texttt{TA\_basic}}
\newcommand{\improved}{\texttt{TA\_improved}}

\usepackage{graphics,color}
\usepackage[colorlinks]{hyperref}
\definecolor{linkblue}{rgb}{0.1,0.1,0.8}

\hypersetup{colorlinks=true,linkcolor=linkblue,filecolor=linkblue,urlcolor=linkblue,citecolor=linkblue}
\usepackage[algo2e,ruled,vlined,linesnumbered]{algorithm2e}

\title{
A New Randomized Algorithm to Approximate the Star Discrepancy 
Based on Threshold Accepting
}
\author{
Michael Gnewuch\thanks{Kiel University, 24098 Kiel, Germany.} 
\and Magnus Wahlstr\"{o}m\thanks{Max-Planck-Institut f{\"u}r Informatik, 66123 Saarbr\"ucken, Germany.}
\and Carola Winzen\footnotemark[2]
}
\date{\today}

\begin{document}
\maketitle

\begin{abstract}
We present a new algorithm for estimating the star discrepancy of 
arbitrary point sets. Similar to the algorithm for discrepancy 
approximation of Winker and Fang [SIAM J.~Numer. Anal. 34 (1997), 
2028--2042] it 
is based on the optimization algorithm threshold accepting. 
Our improvements include, amongst others, a non-uniform sampling 
strategy which is more suited for higher-dimensional inputs, and  
rounding steps which transform axis-parallel boxes, on which the 
discrepancy is to be tested, into 
\emph{critical test boxes}.
These critical test boxes provably yield higher discrepancy values,
and contain the box that exhibits the maximum value of the local 
discrepancy.
We provide comprehensive experiments to test the new algorithm.
Our randomized algorithm computes the exact
discrepancy frequently in all cases where this can be checked (i.e., where the exact discrepancy 
of the point set can be computed in feasible time). 
Most importantly, in higher dimension the new method behaves 
clearly better than all previously known methods.
\end{abstract}

\sloppy{
\section{Introduction}\label{sec:introduction}
Discrepancy theory analyzes the 
irregularity 
of point distributions 
and has considerable theoretical and practical 
relevance. There are many different discrepancy notions with a 
wide range of applications as in optimization, combinatorics,
pseudo random number generation, option pricing,
computer graphics, and other areas, see, e.g., the monographs
\cite{BC87, Cha00, DP10, DT97, FW94, Lem09, Mat99, Nie92, NW10}.

In particular for the important task of multivariate or infinite dimensional
numerical integration, which arises frequently in fields such as finance, statistics, 
physics or quantum chemistry, quasi-Monte
Carlo algorithms relying on low-discrepancy samples have extensively been studied in the last decades.
For several classes of integrands the error of quasi-Monte Carlo approximation 
can be expressed in terms of the discrepancy
of the set of sample points. This is put into a quantitative form by inequalities of Koksma-Hlawka- or 
Zaremba-type, see, e.g., 
\cite{DP10, Gne11, NW10} and the literature mentioned therein. 
The essential point here is that a set of sample points with small discrepancy results 
in a small integration error.

Of particular interest are the star discrepancy
and the weighted star discrepancy, which we define below. 
For theoretical and practical reasons the weighted star discrepancy attracted more and more attention 
over the last few years, see, e.g., \cite{DLP05, HPS08, Joe06, SJ07}. 
In particular, it is very promising for finance applications, see \cite{Slo10}. 

Let $X=(x^i)^n_{i=1}$ be a finite sequence in the $d$-dimensional 
(half-open) unit cube $[0,1)^d$. For $y=(y_1,\ldots,y_d)\in [0,1]^d$
let $A(y,X)$ be the number of points of $X$ lying in the $d$-dimensional
half-open subinterval $[0,y):=[0,y_1)\times\cdots\times[0,y_d)$, and let
$V_y$ be the $d$-dimensional (Lebesgue) volume of $[0,y)$ . We call
\begin{equation*}
d^*_\infty(X) := \sup_{y\in (0,1]^d} \big| V_y - \tfrac{1}{n}
A(y,X) \big|
\end{equation*}
the \emph{$L^\infty$-star discrepancy}, or simply the 
\emph{star discrepancy} of $X$.

For a subset $u\subseteq \{1,\ldots,d\}$ define $\Phi_u: [0,1]^d \to [0,1]^{|u|}, y \mapsto (y_i)_{i\in u}$. For a finite sequence of
non-negative weights $(\gamma_u)_{u\subseteq \{1,\ldots,d\}}$ the \emph{weighted star
discrepancy} of $X$ is defined by 
\begin{equation*}
 d^*_{\gamma, \infty} (X) := \sup_{\emptyset\neq u\subseteq \{1,\ldots,d\}} \gamma_u d^*_\infty 
(\Phi_u(X)).
\end{equation*}
Obviously, the star discrepancy is a special instance of the weighted star
discrepancy.
Further important discrepancy measures are, e.g., the \emph{$L^p$-star discrepancies}
\begin{equation*}
d^*_p(X) := \left( \int_{[0,1]^d} \left| V_y - \frac{1}{n}
A(y,X) \right|^p\,dy \right)^{1/p}, \hspace{2ex} 1 \le p <\infty,
\end{equation*}
and weighted versions thereof. In this article we focus on algorithms to
approximate the star discrepancy, but  
note that these can be used as bases for 
algorithms to approximate the weighted star discrepancy.  

In many applications it is of interest to measure the quality of 
certain sets
by calculating their (weighted or unweighted) star discrepancy, e.g., 
to test whether 
successive pseudo
random numbers are statistically independent \cite{Nie92}, or whether 
given
sample sets are suitable for multivariate numerical integration of 
certain
classes of integrands. 
As explained in 
\cite{DGKP08}, the 
fast calculation or 
approximation of the (weighted) star discrepancy would moreover allow  
efficient randomized semi-constructions of low-discrepancy samples of moderate
size (meaning at most polynomial in the dimension $d$). 
Actually, there are derandomized algorithms known to construct 
such samples deterministically
\cite{DGKP08, DGW09, DGW10}, but these exhibit high running times. 
Therefore, 
efficient semi-constructions would be helpful to avoid the costly 
derandomization procedures.
The critical step in the semi-construction is 
the efficient calculation (or approximation)
of the discrepancy of a randomly chosen set.

The $L^2$-star discrepancy of a given $n$-point set 
in dimension $d$
can be computed with the help of Warnock's formula \cite{War72} with $O(dn^2)$
arithmetic operations. Heinrich and Frank provided an 
asymptotically faster algorithm using $O(n(\log n)^{d-1})$ 
operations for fixed $d$ \cite{FH96, Hei96}. 
Similarly efficient algorithms are not known for the 
star discrepancy (and thus also not for the more general weighted star discrepancy).
In fact it is known that the problem of calculating the star discrepancy of arbitrary
point sets is an $NP$-hard problem \cite{GSW09}. Furthermore, it was shown 
recently that it is also a $W[1]$-hard problem with respect to the parameter $d$ \cite{GKWW10}.
So it is not very surprising that all known algorithms for calculating the star discrepancy
or approximating it up to a user-specified error exhibit running times
exponential in $d$, see \cite{DEM96, Gne08a, Thi01a, Thi01b}. 
Let us have a 
closer look at the problem: 
For a finite sequence 
$X=(x^i)^n_{i=1}$ in $[0,1)^d$ and for $j\in\{1,\ldots,d\}$ we define
\begin{equation*}
\Gamma_j(X) = \{x^i_j \,|\, \ i \in \{1,...,n \} \}
\hspace{2ex}\text{and}\hspace{2ex}
\bar{\Gamma}_j(X) = \Gamma_j(X) \cup \{1\},
\end{equation*}
and the grids
\begin{equation*}
\Gamma(X) = \Gamma_1(X) \times \cdots \times \Gamma_d(X)
\hspace{2ex}\text{and}\hspace{2ex}
\bar{\Gamma}(X) = \bar{\Gamma}_1(X) \times \cdots \times \bar{\Gamma}_d(X).
\end{equation*}
Then we obtain
\begin{equation}
\label{disfor}
d^*_\infty(X) = \max \ \left\{ \ \max_{y \in \bar{\Gamma}(X)} 
\left( V_y - \frac{1}{n} A(y,X) \right)\ ,\ \max_{y \in \Gamma(X)} 
\left(\frac{1}{n} \bar{A}(y,X) - V_y \right) \right\},
\end{equation}
where $\bar{A}(y,X)$ denotes the number of points of $X$ lying in the
closed $d$-dimensional subinterval $[0,y]$. 
(For a proof see \cite{GSW09} or \cite[Thm.~2]{Nie72a}.)
Thus, an enumeration algorithm would provide us with the exact value
of $d^*_\infty(X)$. But since the cardinality of the grid $\Gamma(X)$ for
almost all $X$ is $n^d$, such an algorithm would
be infeasible for large values of $n$ and $d$.

Since no efficient algorithm for the exact calculation 
or approximation 
of the star discrepancy 
up to a user-specified error 
is likely to exist, 
other authors tried to
deal with this large scale integer programming problem by using 
optimization heuristics. In \cite{WF97}, Winker and Fang
used threshold accepting to find lower bounds for the star discrepancy. Threshold accepting \cite{Dueck} is a refined randomized local search 
algorithm based on a similar idea as the simulated annealing algorithm
\cite{KGV83}.
In \cite{Thi01b}, Thi\'emard gave an integer linear programming 
formulation
for the problem and used techniques as cutting plane generation and
branch and bound to tackle it (cf. also \cite{GSW09}).
Quite recently, Shah proposed a genetic algorithm to calculate
lower bounds for the star discrepancy \cite{Sha10}.

Here in this paper we present a new randomized algorithm 
to approximate
the star discrepancy. As the algorithm of Winker and Fang ours is based on 
threshold accepting but adds more problem specific knowledge to it. 
The paper is organized as follows. In Section \ref{sec:previous}
we describe the algorithm of Winker and Fang. In Section \ref{sec:basic}
we present a first version of our algorithm. The most important difference
to the algorithm of Winker and Fang is a new non-uniform sampling strategy that
takes into account the influence of the dimension $d$ and topological characteristics 
of the given point set. In Section 
\ref{sec:improved} we introduce the concept of critical test boxes,
which are the boxes that lead to the largest discrepancy values, including the maximum value.
We present rounding procedures which transform given test boxes into critical test 
boxes. With the help of these procedures and some other modifications, our algorithm 
achieves even better results. 
However, this precision comes at at the cost of larger running times (roughly
a factor two, see Table~\ref{tab:time} in Section~\ref{sec:experiments}). In Section \ref{sec:analysis} we analyze the new sampling strategy and the
rounding procedures in more depth. We provide
comprehensive numerical tests in Section \ref{sec:experiments}. The results indicate 
that our new algorithm is superior to all other known methods, 
especially in higher dimensions. The appendix
contains some technical 
results necessary for our theoretical analyses in Section 
\ref{sec:analysis}.

\section{The Algorithm of Winker and Fang}
\label{sec:previous}
\subsection{Notation}
\label{subsec:notation}
In addition to the notation introduced above, we make use of the following conventions. 

For all positive integers $m\in\N$ we put $[m]:= \{1,\ldots,m\}$. 
If $r\in\R$, let 
$\lfloor r \rfloor := \max\{n\in\Z \setmid n \leq r\}$. 
For the purpose of readability we sometime omit the $\lfloor \cdot \rfloor$ sign, i.e., whenever we write $r$ where an integer is required, we implicitly mean $\lfloor r \rfloor$.

For general $x, y\in [0,1]^d$ we write $x\leq y$ if 
$x_j \leq y_j$ for all $j\in [d]$ and, equivalently, $x <y$ if $x_j < y_j$ for all $j\in [d]$.
The characteristic
function $1_{[0,x)}$ is defined on $[0,1]^d$ by
$1_{[0,x)}(y) := 1$ if $y < x$ and 
$1_{[0,x)}(y) := 0$ otherwise.
We use corresponding conventions for the closed $d$-dimensional box
 $[0,x]$. 

For a given sequence $X=(x^i)_{i=1}^n$ in the $d$-dimensional unit cube $[0,1)^d$, we define the following functions. 
For all $y\in [0,1]^d$ we set
\begin{align*}
&\delta(y) := \delta(y,X) 
:= V_y - A(y,X) 
= V_y -
\frac{1}{n} \sum^n_{k=1} 1_{[0,y)}(x^k) \,,\\
&\bar{\delta}(y) 
:= \bar{\delta}(y,X) 
:= \bar{A}(y,X)-V_y
= \frac{1}{n} \sum^n_{k=1} 1_{[0,y]}(x^k) - V_y\,,
\end{align*}
and $\delta^*(y) := \delta^*(y, X) :=\max\big\{\delta(y), \bar{\delta}(y) \big\}$.
Then $\disc(X)=\max_{y \in \bar{\Gamma}(X)}{\delta^*(y)}$ as discussed in the introduction.

\subsection{The algorithm of Winker and Fang}
\label{subsec:WF}
\label{SUBSEC:WF}

Threshold accepting is an integer optimization heuristic introduced
by Dueck and Scheuer in~\cite{Dueck}. 
Alth\"ofer and Koschnik~\cite{AlK} showed that for suitably chosen parameters, threshold accepting converges to a global optimum if the number 
$I$ of iterations tends to infinity. 
Winker and Fang~\cite{WF97} applied threshold accepting 
to compute the star discrepancy of a given $n$-point configuration.
In the following, we give a short presentation of their algorithm. 
A flow diagram of the algorithm can be found in~\cite{WF97}.

\textbf{Initialization:} The heuristic starts with choosing uniformly at random a starting point 
$x^c\in \bar{\Gamma}(X)$ and calculating $\delta^*(x^c) = \max\{\delta(x^c), 
\overline{\delta}(x^c)\}$. Note that throughout the description of the algorithm, $x^c$ denotes the \emph{currently} used search point.

\textbf{Optimization:}
A number $I$ of iterations is performed. In the $t$-th 
iteration, the algorithm chooses a point $x^{nb}$ uniformly at random from a
given \emph{neighborhood} $\EN(x^c)$ of $x^c$ and calculates $\delta^*(x^{nb})$. 
It then computes
$\Delta \delta^* := \delta^*(x^{nb}) - \delta^*(x^c)$. If $\Delta \delta^* \geq T$ for 
a given (non-positive) threshold value $T$, then $x^c$ is updated, i.e, the algorithm sets $x^c := x^{nb}$.
With the help of the non-positive threshold it shall be avoided to 
get stuck in a bad local maximum $x^c$ of $\delta^*$---``local'' 
with respect to the underlying neighborhood definition. 
The threshold value $T$ changes during the run of the algorithm and ends up at zero.
This should enforce the algorithm to end up at a local
maximum of $\delta^*$ which is reasonably close to $\disc(X)$. 

\textbf{Neighborhood Structure:}
Let us first give the neighborhood definition used in \cite{WF97}. 
For this purpose, let $x \in \bar{\Gamma}(X)$ be given. Let $\ell < n/2$
be an integer and put $k:= 2\ell + 1$.
We allow only a certain number of coordinates to change by 
fixing a value $mc\in [d]$ and choosing $mc$ coordinates 
$j_1,\ldots,j_{mc}\in [d]$ uniformly at random. 
For 
$j \in \{j_1,\ldots,j_{mc}\}$
we consider the set of grid coordinates
\begin{equation*}
\EN_{k,j}(x):=
\Big\{ \gamma\in \bar{\Gamma}_j(X) \Setmid 
\max\{1, \phi^{-1}_j(x_j) - \ell \} \leq \phi^{-1}_j(\gamma)
\leq \min\{|\bar{\Gamma}_j (X)|, \phi^{-1}_j(x_j) + \ell \}\, \Big\},
\end{equation*}
where $\phi_j: [ |\bar{\Gamma}_j(X)| ]\to \bar{\Gamma}_j(X)$ is the ordering of the 
set $\bar{\Gamma}_j(X)$, i.e., $\phi_j(r) < \phi_j(s)$ for $r < s$.
The \emph{neighborhood $\EN^{j_1,\ldots,j_{mc}}_k(x)$ of $x$ of order $k$} is the Cartesian product
\begin{equation}
\label{neigWF}
N^{j_1,\ldots,j_{mc}}_k(x) := \hat{\EN}_{k,1}(x) \times \ldots \times \hat{\EN}_{k,d}(x)\,,
\end{equation}
where $\hat{\EN}_{k,j}(x) = \EN_{k,j}(x)$ for $j\in \{j_1,\ldots,j_{mc}\}$ and
$\hat{\EN}_{k,j}(x) =\{x_j\}$ otherwise. Clearly, $|N^{j_1,\ldots,j_{mc}}_k(x)| \leq (2 \ell +1)^{mc}$.
We abbreviate $\EN^{mc}_{k}(x):=\EN^{j_1,\ldots , j_{mc}}_{k}(x)$ if $j_1,\ldots , j_{mc}$ are 
$mc$ coordinates chosen uniformly at random.

\textbf{Threshold values:} Next, we explain how the threshold sequence is chosen in~\cite{WF97}. The following procedure is 
executed prior to the algorithm itself.
Let $I$ be the total number of iterations to be performed by the algorithm and let $k \leq n$ and $mc \leq d$ be fixed. 
For each $t \in [ \sqrt{I}]$, the procedure computes a pair $(y^t, \tilde{y}^t)$, 
where $y^t \in \bar{\Gamma}(X)$ is chosen uniformly at random
and $\tilde{y}^t \in \EN^{mc}_k(y^t)$, again chosen uniformly at random. 
It then calculates the values $T(t):=-|\delta^*(y^t) - \delta^*(\tilde{y}^t)|$. 
When all values $T(t)$, $t=1,\ldots,\sqrt{I}$, have been computed, the algorithm sorts them in increasing order. For a given $\alpha \in (0.9, 1]$, the $\alpha\sqrt{I}$ values closest to zero are selected as threshold sequence.
The number $J$ of iterations performed for each threshold value is 
$J = \alpha^{-1}\sqrt{I}$.

\section{A First Improved Algorithm -- \basic} 
\label{sec:basic}
Our first algorithm, \basic, builds on the algorithm of winker and Fang as presented in the 
previous section. 
A preliminary, slightly different version of \basic\ can be found in~\cite{Win07}. This version was used in~\cite{DGW10} to provide lower bounds for the comparison of the star discrepancies of different point 
sequences. In particular in higher dimensions it performed better than any other method tested by the authors.

Recall that the algorithm of winker and Fang employs a uniform probability distribution on 
$\bar{\Gamma}(X)$ and the neighborhoods $\EN^{mc}_{k}(x)$ for all random decisions. 

Firstly, this is not appropriate for higher-dimensional inputs: In any dimension $d$ it is most
likely that the discrepancy of a set $X$ is caused by test boxes with volume at least $c$,
$c$ some constant in $(0,1)$. Thus in higher dimension $d$ we expect the upper right corners
of test boxes with large local discrepancy to have coordinates at least $c^{1/d}$.
Thus it seems appropriate for higher dimensional sets $X$ to increase the weight of those points in the grid $\bar{\Gamma}(X)$ with larger coordinates whereas we decrease the weight of the points with small coordinates.

Secondly, a uniform probability distribution does not take into account
the topological characteristics of the point set $X$ as, e.g., distances between the points in
the grid $\bar{\Gamma}(X)$: If there is a grid cell $[x,y]$ in $\bar{\Gamma}(X)$
(i.e., $x,y\in \bar{\Gamma}(X)$ and $\phi^{-1}_j(y_j) = \phi^{-1}_j(x_j)+1$ for all 
$j\in [d]$, where $\phi_j$ is again the ordering of the set $\bar{\Gamma}_j(X)$)
with large volume, we would expect that $\bar{\delta}(x)$ or $\delta(y)$ are also 
rather large.

Thus, on the one hand, it seem better to consider a modified probability measure on 
$\bar{\Gamma}(X)$ which accounts for the influence of the dimension and the topological
characteristics of $X$. On the other hand, if $n$ and $d$ are large, we clearly cannot
afford an elaborate precomputation of the modified probability weights.

To cope with this, the non-uniform sampling strategy employed by \basic\
consists of two steps: 
\begin{itemize}
\item A continuous sampling step, where we select a point in the whole
$d$-dimensional unit cube (or in a ``continuous'' neighborhood of $x^c$)  
with respect to a non-uniform (continuous) probability measure $\pi^d$, which
is more concentrated in points with larger coordinates.
\item A rounding step, where we round the selected point to the grid  $\bar{\Gamma}(X)$.
\end{itemize}
In this way we address both the influence of the dimension and the topological
characteristics of the point set $X$. This works without performing any precomputation of
probability weights on $\bar{\Gamma}(X)$ -- instead, the random generator, the change of measure
on $[0,1]^d$ from the $d$-dimensional Lebesgue measure to $\pi^d$, and our rounding procedure
do this implicitly!  
Theoretical and experimental justifications for our non-uniform sampling strategy 
can be found in Section~\ref{sec:analysis} and Section~\ref{SEC:EXPERIMENTS}.

\subsection{Sampling of Neighbors}
\label{subsec:neighbors}
\label{SUBSEC:NEIGHBORS}

In the following, we present how we modify the probability distribution over the neighborhood sets. 
Our non-uniform sampling strategy consists of the following two steps.

\textbf{Continuous Sampling}
Consider a point 
$x\in \bar{\Gamma}(X)$. For fixed $mc\in [d]$ let $j_1,\ldots , j_{mc} \in [d]$ be pairwise different coordinates. For $j \in \{j_1,\ldots , j_{mc}\}$ 
let $\varphi_j: [ |\bar{\Gamma}_j(X)\cup \{0\}| ]\to \bar{\Gamma}_j(X)\cup \{0\}$ be the ordering 
of the 
set $\bar{\Gamma}_j(X)\cup \{0\}$ (in particular $\varphi_j(1) = 0$).
Let us now consider the real interval $C_{k,j}(x) := [\xi(x_j), \eta(x_j)]$
with
\begin{align*}
 \xi(x_j) := \varphi \big( \max\{1, \varphi^{-1}(x_j) -\ell\} \big)
\text{ and }
 \eta(x_j) := \varphi \big( \min\{|\bar{\Gamma}_j(X)\cup \{0\}|, 
\varphi^{-1}(x_j) +\ell\} \big) \,.
\end{align*}
Our new \emph{neighborhood $C^{j_1,\ldots , j_{mc}}_k(x)$ of $x$ of order $k$} is the Cartesian product
\begin{equation}
\label{neigGK}
C^{j_1,\ldots , j_{mc}}_k(x) := \hat{C}_{k,1}(x) \times \ldots \times \hat{C}_{k,d}(x)\,,
\end{equation}
where $\hat{C}_{k,j}(x) = C_{k,j}(x)$ for $j\in \{j_1,\ldots,j_{mc}\}$ and
$\hat{C}_{k,j}(x) =\{x_j\}$ otherwise. 
We abbreviate $C^{mc}_{k}(x):=C^{j_1,\ldots , j_{mc}}_{k}(x)$ if $j_1,\ldots , j_{mc}$ are 
$mc$ coordinates chosen uniformly at random.

Instead of endowing $C^{j_1,\ldots , j_{mc}}_{k}(x)$ with the Lebesgue measure on the non-trivial components, we choose a different probability distribution which we describe in the following. 
First, let us consider the polynomial product measure
\begin{equation*}
\pi^d(\,dx) = \otimes^d_{j=1}  f(x_j)\,\lambda(\,dx_j) 
\text{ with density function }
f: [0,1] \to \R\,, \,r \mapsto d r^{d-1}
\end{equation*} on $[0,1]^d$; here $\lambda = \lambda^1$ should denote the one-dimensional 
Lebesgue measure.
Notice that in dimension $d=1$ we have $\pi^1 = \lambda$.
Picking a random point $y\in [0,1]^d$ with respect to the new 
probability measure
$\pi^d$ can easily be done in practice by sampling a point
$z\in [0,1]^d$ with respect to $\lambda^d$ and then putting 
$y:= (z_1^{1/d},\ldots,z^{1/d}_d)$.

We endow $C^{j_1,\ldots , j_{mc}}_{k}(x)$ with the 
probability distribution induced by the polynomial product measure  
on the $mc$ non-trivial components $C_{k,j_1}(x),\ldots, C_{k,j_{mc}}(x)$. 
To be more 
explicit, we map each 
$C_{k,j}(x)$, $j\in \{j_1,\ldots, j_{mc}\}$,
to the unit interval $[0,1]$ by
\begin{equation*}
\Psi_j : C_{k,j}(x) \to [0,1], 
r \mapsto \frac{r^d - (\xi(x_j))^d}{(\eta(x_j))^d - (\xi(x_j))^d}\,.
\end{equation*}
Recall that $\xi(x_j) := \min C_{k,j}(x)$ and $\eta(x_j) := \max C_{k,j}(x)$. The inverse mapping $\Psi^{-1}_j$ is then given by
\begin{equation*}
\Psi^{-1}_j : [0,1] \to C_{k,j}\,, \, 
s\mapsto \Big( \big((\eta(x_j))^d - (\xi(x_j))^d \big) s + (\xi(x_j))^d \Big)^{1/d}\,.
\end{equation*}
If we want to sample a random point $y\in C^{j_1,\ldots , j_{mc}}_{k}(x)$, we randomly choose 
scalars $s_1,\ldots, s_{mc}$ in $[0,1]$ with respect to $\lambda$ and put
$y_{j_i} := \Psi^{-1}_{j_i}(s_i)$ for 
$i = 1,\ldots, mc$. For indices $j \notin \{j_1,\ldots, j_{mc}\}$
we set $y_j := x_j$. 

\textbf{Rounding Procedure:} We round the point $y$
once up and once down to the nearest points $y^+$ and $y^-$ in 
$\bar{\Gamma}(X)$. 
More precisely, for all $j \in [d]$, let
$y^+_j:=\min\{x^i_j \in \bar{\Gamma}_j(X) \setmid y_j \leq x^i_j \}$. 
If $y_j \geq \min{\bar{\Gamma}_j(X)}$ we set
$y^-_j:=\max\{x^i_j \in \bar{\Gamma}_j(X) \setmid y_j \geq x^i_j \}$ and in case $y_j < \min{\bar{\Gamma}_j}(X)$, we set $y^-_j:= \max{\Gamma_j}(X)$.

Obviously, 
$A(y^+,X) =  A(y,X)$ and thus,
$\delta(y^+)= V_{y^+} - A(y^+,X) \geq V_{y} - A(y,X) = \delta(y)\,.$
Similarly, if $y_j \geq \min \Gamma_j(X)$ for all $j \in [d]$, we have $\bar{A}(y^-,X) = \bar{A}(y,X)$. 
Hence,
$\bar{\delta}(y^-)= \bar{A}(y^-,X)- V_{y^-} \geq \bar{A}(y,X)- V_{y} = \bar{\delta}(y)\,.$
If $y_j < \min \Gamma_j(X)$ for at least one $j \in [d]$
we have $\bar{\delta}(y) \le 0$ since $\bar{A}(y,X)=0$. 
But we also have $A(y,X)=0$ and thus
$\delta^*(y) = \delta (y) \leq \delta(y^+)\,.$
Putting everything together, we have shown that 
$\max \{ \delta(y^+), \bar{\delta}(y^-) \} \geq \delta^*(y)\,.$ 

Since it is only of insignificant additional computational cost to also compute 
$\bar{\delta}(y^{-,-})$ where $y^{-,-}_j:=y^-_j$ for all $j \in [d]$ with $y_j \geq \min{\bar{\Gamma}_j(X)}$ and $y^{-,-}_j:=\min{\bar{\Gamma}_j(X)}$ for $j$ with $y_j < \min{\bar{\Gamma}_j(X)}$, we also do that in case at least one such $j$ with $y_j < \min{\bar{\Gamma}_j(X)}$ exists.

For sampling a neighbor $x^{nb}$ of $x^c$ the algorithm thus does the following. 
First, it samples $mc$ coordinates $j_1,\ldots, j_{mc} \in [d]$ uniformly at random. 
Then it samples a point $y \in C^{j_1,\ldots , j_{mc}}_{k}(x^c)$ as described above, computes the rounded grid points $y^+$, $y^-$, and $y^{-,-}$ and computes the discrepancy $\delta_{\Gamma}^*(y):=\max\{\delta(y^+), \bar{\delta}(y^-), \bar{\delta}(y^{-,-})\}$ of the rounded grid points. The subscript $\Gamma$ shall indicate that we consider the rounded grid points.
As in the algorithm of Winker and Fang, \basic\ updates $x^c$ if and only if 
$\Delta \delta^* = \delta_{\Gamma}^*(y) - \delta^*(x^c) \geq T$, where $T$ denotes the current threshold. In this case we always update $x^c$ with the best rounded test point, i.e., we update $x^c:=y^+$ if $\delta_{\Gamma}^*(y) = \delta(y^+)$, $x^c:=y^-$ if $\delta_{\Gamma}^*(y) = \bar{\delta}(y^-)$, and $x^c:=y^{-,-}$ otherwise. 

\subsection{Sampling of the Starting Point}
\label{subsec:starting}
Similar to the probability distribution on the neighborhood sets, we 
sample the starting point $x^c$ as follows. 
First, we sample a point $x$ from $[0,1]^d$ according to $\pi^d$. We then round $x$ up and down to $x^+$, $x^-$, and $x^{-,-}$, respectively and again we set $x^c:=x^+$ if $\delta_{\Gamma}^*(x) = \delta(x^+)$, $x^c:=x^-$ if $\delta_{\Gamma}^*(x) = \bar{\delta}(x^-)$, and we set $x^c:=x^{-,-}$ otherwise. 

\subsection{Computation of Threshold Sequence}
\label{subsec:threshold}
The modified neighborhood sampling is also used for computing the sequence of threshold values. 
If we want the algorithm to perform $I$ iterations, we compute the threshold sequence as follows. 
For each $t \in [\sqrt{I}]$ we sample a pair
$(y^t, \tilde{y}^t)$, where $y^t \in \bar{\Gamma}(X)$ is sampled as is the starting point 
and $\tilde{y}^t \in \bar{\Gamma}(X)$ is a neighbor of $y^t$, sampled according to the procedure described in Section~\ref{subsec:neighbors}.
The thresholds $-|\delta^*(y^t) - \delta^*(\tilde{y}^t)|$ are sorted in increasing order and each threshold will be used for $\sqrt{I}$ iterations of \basic. 
Note that by this choice, we are implicitly setting $\alpha:=1$ in the notion of the algorithm of Winker and Fang.

\section{Further Improvements -- Algorithm \improved}
\label{sec:improved}
In the following, we present further modifications which we applied to the basic algorithm \basic. We call the new, enhanced algorithm \improved.
 
The main improvements, which we describe in more detail below, are 
\textbf{(i)} a further reduction of the search space by introducing new rounding procedures (``snapping''),
\textbf{(ii)} shrinking neighborhoods and growing number of search directions,
and
\textbf{(iii)} separate optimization of $\delta$ and $\bar{\delta}\,.$

\subsection{Further Reduction of the Search Space}
\label{subsec:snapping}
\label{SUBSEC:SNAPPING}

We mentioned that for calculating the star discrepancy it is sufficient to test 
just the points $y\in \bar{\Gamma}(X)$ and to calculate $\delta^*(y)$, cf. equation~(\ref{disfor}). 
Therefore $\bar{\Gamma}(X)$ has been
the search space we have considered so far. 
But it is possible to reduce the 
cardinality of the search space even further.
 
We obtain the reduction of the search space via a rounding procedure which we call \emph{snapping}. As this is an important element in the modified algorithm, we now discuss the underlying concept of critical
points (or test boxes).
For $y\in [0,1]^d$ we define
\begin{equation*}
S_j(y) := \prod^{j-1}_{i=1}[0,y_i) \times \{y_j\}
\times \prod^d_{k= j+1}[0,y_k)\,,
\hspace{1ex}j=1, \ldots, d\,.
\end{equation*}
We say that $S_j(y)$ is a \emph{$\delta(X)$-critical surface} if $S_j(y)\cap
\{x^1,\ldots,x^n\} \neq \emptyset$ or $y_j = 1$.
We call $y$ a \emph{$\delta(X)$-critical point} if for all $j\in [d]$ the surfaces 
$S_j(y)$ are critical. Let $\Ap$ denote the set of $\delta(X)$-critical 
points in $[0,1]^d$. 

Let $\bar{S}_j(y)$ be the closure of $S_j(y)$, i.e., 
\begin{equation*}
\bar{S}_j(y) := \prod^{j-1}_{i=1}[0,y_i] \times \{y_j\}
\times \prod^d_{k= j+1}[0,y_k]\,,
\hspace{1ex}j=1, \ldots, d\,.
\end{equation*}
We say $\bar{S}_j(y)$ is a \emph{$\bar{\delta}(X)$-critical surface} if 
$\bar{S}_j(y)\cap
\{x_1,\ldots,x_n\} \neq \emptyset$.
If for all $j\in [d]$ the surfaces $\bar{S}_j(y)$ are 
$\bar{\delta}(X)$-critical, then we
call $y$ a \emph{$\bar{\delta}(X)$-critical point}. 
Let $\bar{\Ap}$ denote the set of $\bar{\delta}(X)$-critical 
points in $[0,1]^d$. 
We call $y$ a \emph{$\delta^*(X)$-critical point} if 
$y\in \Ap^* := \Ap \cup \bar{\Ap}$. 

For $j\in [d]$ let $\nu_j := |\bar{\Gamma}_j(X)|$, and let again 
$\phi_j: [\nu_j] \to \bar{\Gamma}_j(X)$ denote the ordering of $\bar{\Gamma}_j(X)$. 
Let $\Phi: [\nu_1]\times \cdots \times [\nu_d] \to \bar{\Gamma}(X)$ be the mapping
with components $\phi_j$, $j=1,\ldots,d$. We say that a multi-index
$(i_1,\ldots,i_d)\in [n+1]^d$ is a \emph{$\delta(X)$-critical multi-index} 
if $\Phi(i_1,\ldots, i_d)$ is a $\delta(X)$-critical point. We use 
similar definitions in cases where we deal with $\bar{\delta}(X)$
or $\delta^*(X)$.

\begin{lemma}
\label{Kritisch}
Let $X=\{x^1,\ldots,x^n\}$ be a $n$-point configuration in $[0,1)^d$.
Let $\Ap = \Ap(X)$, $\bar{\Ap} = \bar{\Ap}(X)$, and 
$\Ap^* = \Ap^*(X)$ be as defined above.
Then $\Ap$, $\bar{\Ap}$ and $\Ap^*$ are non-empty subsets of $\bar{\Gamma}(X)$.
Furthermore, 
\begin{equation*}
\sup_{y\in [0,1]^d} \delta(y) = \max_{y\in\Ap} \delta(y)\,,
\sup_{y\in [0,1]^d} \bar{\delta}(y) 
= \max_{y\in\bar{\Ap}} \bar{\delta}(y)\\
\hspace{1ex}\text{and}\hspace{1ex}
\sup_{y\in[0,1]^d}\delta^*(y) = \max_{y\in \Ap^*} \delta^*(y).
\end{equation*}
\end{lemma}

\begin{proof}
The set $\Ap$ is not empty, since it contains the point $(1,\ldots,1)$.
Let $y \in \Ap$. 
By definition, we find for all $j\in [d]$ an index $\sigma(j)
\in [n]$ with $y_j = x^{\sigma(j)}_j$ or we have $y_j = 1$. Therefore
$y\in\bar{\Gamma}(X)$.
Let $z\in [0,1]^d\setminus \Ap$. Since $\delta(z) = 0$ if $z_j = 0$
for any index $j$, we may assume $z_j > 0$ for all $j$. 
As $z \notin \Ap$ there exists a $j\in [d]$ 
where $S_j(z)$ is not
$\delta(X)$-critical. In particular, we have $z_j < 1$. 
Let now $\tau\in \bar{\Gamma}_j(X)$ be the smallest value with
$z_j < \tau$. Then the point 
$\hat{z} := (z_1,\ldots, z_{j-1}, \tau, z_{j+1}, \ldots, z_d)$
fulfills $V_{\hat{z}} > V_z$. 
Furthermore, the sets $[0, \hat{z}) \setminus
[0, z)$ and $X$ are disjoint. So $[0, \hat{z})$ and 
$[0,z)$ contain the same points of $X$. 
In particular we have $A(\hat{z},X)=A(z,X)$ and thus,
$\delta(\hat{z}) > \delta(z)$. 
This argument verifies $\sup_{y\in [0,1]^d} \delta(y) = \max_{y\in
  \Ap} \delta(y)$. The remaining statements of Lemma 
\ref{Kritisch} can be proven with similar simple arguments.
\end{proof}

We now describe how to use this concept in our algorithm. Let us first describe how we sample a neighbor $x^{nb}$ of a given point $x^c$. The procedure starts exactly as described in Section~\ref{subsec:neighbors}. That is, we first sample $mc$ coordinates $j_1,\ldots, j_{mc} \in [d]$ uniformly at random.
Next, we sample $y\in C^{j_1,\ldots,j_{mc}}_k(x^c)$ according to the probability distribution induced by the polynomial product measure $\pi^d$ on the non-trivial components of $C^{j_1,\ldots,j_{mc}}_k(x^c)$, cf. Section~\ref{SUBSEC:NEIGHBORS}. Again we round $y$ up and down to the nearest grid points $y^+$, $y^-$ and $y^{-,-}$, respectively.
We then apply the following snapping procedures\footnote{The snapping procedure is the same for $y^-$ and $y^{-,-}$. Therefore, we describe it for $y^-$ only.}. 

\textbf{Snapping down.}
We aim at finding a $\bar{\delta}(X)$-critical point $y^{-,sn} \leq y^-$ such that the closed 
box $[0,y^{-,sn}_j]_{j=1}^d$ contains exactly the same points of $X$ as the box $[0,y^{-}_j]_{j=1}^d$. 
We achieve this by simply setting for all $j \in [d]$ 
$$
y^{-,sn}_j := \max \{ x^i_j \setmid i \in [n], x^i \in [0,y^{-}] \}\,.
$$
From the algorithmic perspective, we initialize $y^{-,sn}:=(0,\ldots,0)$ and check for each index $i \in [n]$ whether $x^i \in [0,y^{-}]$. If so, we check for all $j \in [d]$ whether $x^i_j \leq y^{-,sn}_j$ and update $y^{-,sn}_j:=x^i_j$ otherwise.

\textbf{Snapping up\footnote{Being aware that ``snapping up'' is an oxymoron, we still use this notation as it eases readability in what follows.}.} 
Whereas snapping down was an easy task to do, the same is not true for \emph{snapping up}, i.e., rounding a point to a $\delta(X)$-critical one. More precisely, given a point $y^+$, there are multiple $\delta(X)$-critical points $y^{+,sn} \geq y^+$ such that the open box created by $y^{+,sn}$ contains only those points which are also contained in $[0,y^+)$. 

Given that we want to perform only one snapping up procedure per iteration, we use the following random version of snapping upwards.
In the beginning, we initialize $y^{+,sn}:=(1,\ldots,1)$. Furthermore, we pick a permutation 
$\sigma$ of $[d]$ uniformly at random from the set $S_d$ of all permutations of set $[d]$.
For each point $x \in \{x^i \setmid i \in [n] \}$ we now do the following. 
If $x \in [0,y^{+})$ or $x_j\geq y_j^{+,sn}$ for at least one $j \in [d]$, we do nothing. Otherwise we update $y^{+,sn}_{\sigma(j)}:=x_{\sigma(j)}$ for the smallest $j\in [d]$ with $x_{\sigma(j)} \geq y^+_{\sigma(j)}$. After this update, $x$ is no longer inside the open box generated by $y^{+,sn}$.

Note that snapping up is subject to randomness as the $\delta(X)$-critical point obtained by our snapping procedure can be different for different permutations $\sigma \in S_d$.

The complexity of both snapping procedures is of order $O(nd)$. In our experiments, the snapping procedures caused a delay in the (wall clock) running time by a factor of approximately two, if compared to the running time of \basic. It is not difficult to verify the following.

\begin{lemma}
Let $X$ be a given $n$-point sequence in $[0,1)^d$
For all $y \in [0,1]^d$, the point $y^{+,sn}$, computed as described above, is $\delta(X)$-critical and both $y^{-,sn}$ and $y^{-,-,sn}$ are $\bar{\delta}(X)$-critical. 
\end{lemma}

In the run of the algorithm we now do the following. Given that we start in some grid point $x^c$, we sample $y\in C^{mc}_k(x^c)$ and we round $y$ to the closest grid points $y^+, y^-, y^{-,-} \in \bar{\Gamma}(X)$ as described in Section~\ref{subsec:neighbors}. Next we compute the $\delta(X)$-critical point $y^{+,sn}$, the $\bar{\delta}(X)$-critical point $y^{-,sn}$, and, if $y^-\neq y^{-,-}$, we also compute the $\bar{\delta}(X)$-critical point $y^{-,-,sn}$.
We decide to update $x^c$ if $\Delta \delta^* = \delta^{*,sn}(y) - \delta^{*,sn}(x^c) \geq T$, where $T$ is the current threshold, $\delta^{*,sn}(y):=\max\{\delta(y^{+,sn}), \bar{\delta}(y^{-,sn}), \bar{\delta}(y^{-,-,sn})\}$, and $\delta^{*,sn}(x^c)$ is the value as was computed in the iteration where $x^c$ was updated last.
Note that we do not update $x^c$ with any of the critical points $y^{+,sn}$, $y^{-,sn}$, or $y^{-,-,sn}$ but only replace $x^c$ with the simple rounded grid points $y^{+}$, $y^{-}$, or $y^{-,-}$, respectively. More precisely, we update $x^c:=y^+$ if $\delta^{*,sn}(y) = \delta(y^{+,sn})$, $x^c:=y^-$ if $\delta^{*,sn}(y) = \bar{\delta}(y^{-,sn})$, and $x^c:=y^{-,-}$, otherwise. 

\subsubsection{Computation of the Starting Point and the Threshold Sequence}  
When computing the starting point $x^c$ we first sample a random point $x$ from $[0,1]^d$ according to $\pi^d$ (see Section~\ref{subsec:neighbors}) and compute $x^+$ and $x^-$, and, if applicable, $x^{-,-}$. We also compute the $\delta(X)$- and $\bar{\delta}(X)$-critical points $x^{+,sn}$, $x^{-,sn}$, and $x^{-,-,sn}$ and set $\delta^{*,sn}(x):=\max\{\delta(x^{+,sn}), \bar{\delta}(x^{-,sn}), \bar{\delta}(x^{-,-,sn})\}$. We put $x^c:=x^+$ if $\delta^{*,sn}(x) = \delta(x^{+,sn})$, $x^c:=x^-$ if $\delta^{*,sn}(x) = \bar{\delta}(x^{-,sn})$, and $x^c:=x^{-,-}$, otherwise.

For computing the threshold sequence, we also use the $\delta(X)$- and $\bar{\delta}(X)$-critical $\delta^{*,sn}$-values. That is, for $t=1,\ldots,\sqrt{I}$ we compute $t$-th pair $(y^t, \tilde{y}^t)$ by first sampling a random starting point $y^t$ as described above (i.e., $y^t \in \{x^+, x^-, x^{-,-}\}$ for some $x$ sampled from $[0,1]^d$ according to $\pi^d$
and $y^t=x^+$ if $\delta^{*,sn}(x) = \delta(x^{+,sn})$, $y^t=x^-$ if $\delta^{*,sn}(x) = \bar{\delta}(x^{-,sn})$, and $y^t=x^{-,-}$ otherwise). 
We then compute a neighbor $\tilde{y}^t \in C^{mc}_k(y^t)$ and the maximum of the discrepancy of the $\delta(X)$- and $\bar{\delta}(X)$-critical points $\delta^{*,sn}(\tilde{y}^t):=\max\{\delta(\tilde{y}^{t,+,sn}), \bar{\delta}(\tilde{y}^{t,-,sn}), \bar{\delta}(\tilde{y}^{t,-,-,sn})\}$. 
Finally, we sort the threshold values $T(t):=- |\delta^{*,sn}(y^t)-\delta^{*,sn}(\tilde{y}^t)|, t=1,\ldots,\sqrt{I}$ in increasing order. This will be our threshold sequence.

\subsection{Shrinking Neighborhoods and Growing Number of Search Directions}
\label{subsec:shrinking}

We add the concept of shrinking neighborhoods, i.e., we consider neighborhoods that decrease in size during the run of the algorithm.
The intuition here is the following. In the beginning, we want the algorithm to make large jumps. This allows it to explore different regions of the search space. 
However, towards the end of the algorithm we want it to become more local, allowing it to explore large parts of the local neighborhood. We implement this idea by iteratively shrinking the $k$-value. At the same time, we increase the $mc$-value, letting the algorithm explore the local neighborhood more thoroughly. 

More precisely, we do the following. 
In the beginning we set $\ell:=(n-1)/2$ and $mc:=2$. That is, the algorithm is only allowed to change few coordinates of the current search point but at the same time it can make large jumps in these directions. Recall that $k=2\ell+1$.
In the $t$-th iteration (out of a total number of $I$ iterations) we then update 
\begin{equation*}
\ell:=\frac{n-1}{2} \cdot \frac{I-t}{I} + \frac{t}{I}\, \text{  and  }\, mc:=2+\frac{t}{I} (d-2)\,.
\end{equation*}

For the computation of the threshold sequence, we equivalently scale $k$ and $mc$ by initializing $\ell:=(n-1)/2$ and $mc:=2$ and then setting for the computation of the $t$-th pair $(y^t, \tilde{y}^t)$
\begin{equation*}
\ell:=\frac{n-1}{2} \cdot \frac{\sqrt{I}-t}{\sqrt{I}} + \frac{t}{\sqrt{I}}\, \text{  and  }\, mc:=2+\frac{t}{\sqrt{I}} (d-2)\,.
\end{equation*}
Recall that we compute a total number of $\sqrt{I}$ threshold values.

\subsection{Seperate Optimization of $\delta$ and $\bar{\delta}$}
\label{subsec:split}

Our last modification is based on the intuition that the star discrepancy is either obtained by an open, subproportionally filled box (i.e., there exists a $y \in \bar{\Gamma}(X)$ such that $d^*_\infty(X) = \delta(y)$), in which case one might assume that there are many points $\tilde{y}$ with large $\delta(\tilde{y})$-values. 
Alternatively, if the discrepancy is obtained by a closed, overproportionally filled box (i.e., there exists a $y \in \bar{\Gamma}(X)$ such that $d^*_\infty(X) = \bar{\delta}(y)$), we assume that there are multiple such points $\tilde{y}$ with large $\bar{\delta}(\tilde{y})$-values. 
This intuition triggered us to test also the following split variant of the algorithm. 

In the $\delta$-version of the algorithm, we only consider open test boxes. That is, whenever we want to sample a random starting point [a random neighbor], we proceed exactly as described in Section~\ref{subsec:snapping} but instead of computing both $y^+$ and $y^-$ (and, potentially $y^{-,-}$) as well as the $\delta(X)$- and $\bar{\delta}(X)$-critical points $y^{+,sn}$, $y^{-,sn}$, and $y^{-,-,sn}$ in the notation of Section~\ref{subsec:snapping}, we only compute $y^+$ [and $y^{+,sn}$], and we initialize $x^c:=y^+$ [we update $x^c:=y^+$ if and only if $\Delta \delta = \delta(y^{+,sn}) - \delta\big((x^c)^{+,sn}\big) \geq T$, where $T$ again denotes the current threshold].

The $\bar{\delta}$-version is symmetric. We compute both $y^-$ and $y^{-,-}$ as well as the $\bar{\delta}(X)$-critical points $y^{-,sn}$ and $y^{-,-,sn}$, and we initialize $x^c:=y^-$ or $x^c:=y^{-,-}$ [we update $x^c:=y^-$ or $x^c:=y^{-,-}$ if and only if $\Delta \bar{\delta} = \max\{\bar{\delta}(y^{-,sn}), \bar{\delta}(y^{-,-,sn}) \} - \bar{\delta}\big((x^c)^{-,sn}\big) \geq T$].

Note that only $\delta$-values (or $\bar{\delta}$-values, respectively) are considered for the computation of the threshold sequence as well.

The algorithm is now the following. We perform $I$ iterations of the $\delta$-version of the algorithm and $I$ iterations of the $\bar{\delta}$-version. We then output the maximum value obtained in either one of the two versions. 

It should be noted that a large proportion of the computational cost of \improved\ lies in the snapping procedures.
Thus, running $I$ iterations of the $\delta$-version followed by $I$ iterations of the $\bar{\delta}$-version has a comparable running time to running $I$ iterations of an algorithm  of the ``mixed'' form where we snap each point up and down to the $\delta(X)$- and $\bar{\delta}(X)$-critical grid points.
Furthermore, as most modern CPUs are multicore and able to run several programs in parallel, the actual wall-clock cost of switching from \basic\ to the split version of \improved\ may be smaller still.

Algorithm~\ref{alg:improved} summarizes \improved. Note that $\bar{\delta}(n,d,X,I)$ is the equivalent of Algorithm~\ref{alg:delta} where we replace $\delta$ by $\bar{\delta}$, $x^+$ by $x^-$ etc. The same is true for Subroutine $\thresh(n,d,X,\sqrt{I},\bar{\delta})$ for computing the threshold sequence for the $\bar{\delta}$-version.

\begin{algorithm2e}[h!]
\textbf{Input:}\\
	\quad Problem instance: $n \in \N$, $d \in \N$, sequence $X=(x^i)_{i=1}^n$ in $[0,1)^d$.\\
	\quad Number of iterations $I$.\\
\textbf{Computation of a lower bound for $\disc(X)$:}\\
	\quad $\delta:=\delta(n,d,X,I)$ /* Output of $I$ iterations of the $\delta$-version.*/\\
	\quad $\bar{\delta}:=\bar{\delta}(n,d,X,I)$ /* Output of $I$ iterations of the 	$\bar{\delta}$-version.*/\\
\textbf{Output:} $\delta^*:=\max \{\delta, \bar{\delta} \}$.
\caption{The algorithm \improved\ for computing lower bounds of the star discrepancy $\disc(X)$.}
\label{alg:improved}
\end{algorithm2e}

\begin{algorithm2e}[h!]
\textbf{Initialization:}\\
\Indp $TS=(T(i))_{i=1}^{\sqrt{I}}:=\thresh(n,d,X,\sqrt{I},\delta)$ /*Compute the threshold sequence of length $\sqrt{I}$.*/\\
Sample the starting point: pick $x \in [0,1)^d$ with respect to $\pi^d$ and round $x$ up to the nearest grid point $x^+$. Compute the $\delta(X)$-critical point $x^{+,sn}$ and $\delta(x^{+,sn})$.\\
Initialize $x^c:=x^+$, $\globalbest:=\delta(x^{+,sn})$, $\current:=\delta(x^{+,sn})$,
$\ell:=\lfloor\tfrac{n-1}{2}\rfloor$, and $mc:=2$.\\
\Indm \For{$i=1,\ldots, \sqrt{I}$}
{
Update threshold value $T:=T(i)$.
\For{$t=(i-1)\sqrt{I}+1,\ldots, (i-1)\sqrt{I}+\sqrt{I}$}
{
Update $\ell:=\lfloor\frac{n-1}{2} \cdot \frac{I-t}{I} + \frac{t}{I}\rfloor$ and $mc:=2+\lfloor\frac{t}{I}\rfloor (d-2)$.\\
Sample $y \in C^{mc}_k(x^c)$ as described in Section~\ref{SUBSEC:NEIGHBORS}.\\
Round $y$ up to the nearest grid point $y^+ \in \bar{\Gamma}(X)$ 
and compute the $\delta(X)$-critical point $y^{+,sn}$ as well as $\delta(y^{+,sn})$.\\
\lIf{$\delta(y^{+,sn})>\globalbest$}{update $\globalbest:=\delta(y^{+,sn})$}.\\
\lIf{$\Delta \delta := \delta(y^{+,sn}) - \current \geq T$}{update $x^c:=y^+$ and $\current:=\delta(y^{+,sn})$}.
}}\caption{The $\delta$-version $\delta(n,d,X,I)$.}
\label{alg:delta}
\end{algorithm2e}

\begin{algorithm2e}[h!]
\textbf{Initialization:}\\
\Indp Initialize 
$\ell:=\lfloor\tfrac{n-1}{2}\rfloor$, and $mc:=2$.\\
\Indm 
\For{$t=1,\ldots, \sqrt{I}$}
{
Update 
$\ell:=\lfloor\frac{n-1}{2} \cdot \frac{\sqrt{I}-t}{\sqrt{I}} + \frac{t}{\sqrt{I}}\rfloor$ and 
$mc:=2+\lfloor t/\sqrt{I}\rfloor (d-2)$.\\
Sample a random point: pick $x \in [0,1)^d$ with respect to $\pi^d$ and round $x$ up to the nearest grid point $x^+$. Compute the $\delta(X)$-critical point $x^{+,sn}$ and $\delta(x^{+,sn})$.\\
Sample $y \in C^{mc}_k(x^+)$ as described in Section~\ref{SUBSEC:NEIGHBORS}.\\
Round $y$ up to the nearest grid point $y^+ \in \bar{\Gamma}(X)$ 
and compute the $\delta(X)$-critical point $y^{+,sn}$ as well as $\delta(y^{+,sn})$.\\
$\tilde{T}(i):=-|\delta(y^{+,sn}) - \delta(x^{+,sn})|$.
}Sort thresholds in increasing order to obtain threshold sequence $(T(i))_{i=1}^n$ with $T(i)\leq T(i+1)$ for all $i \in [n-1]$.
\caption{Subroutine $\thresh(n,d,X,\sqrt{I},\delta)$ for computing the threshold sequence.}
\label{alg:thresh}
\end{algorithm2e}

\subsection{Further Variants of the Algorithm} 
\label{sssec:variant}
We do not update $x^c$ with the critical points, since our experiments showed that the performance of the 
algorithm can be significantly improved by updating with the (simply) rounded, not necessarily critical points. This seems to allow the algorithm more flexibility and prevents it from getting stuck in a local optimum too early. 

We also tested a variant of the algorithm where we only update the best-so-far solution with $\delta^{*,sn}(y):=\max\{\delta(y^{+,sn}), \bar{\delta}(y^{-,sn}), \bar{\delta}(y^{-,-,sn})\}$ but where all other decisions are only based on the value $\delta_{\Gamma}^*(y):=\max\{\delta(y^+), \bar{\delta}(y^-), \bar{\delta}(y^{-,-})\}$. 
That is, this algorithm does exactly the same as \basic\ but in addition computes the $\delta(X)$- and $\bar{\delta}(X)$-critical points and stores the largest values of $\delta^{*,sn}$. 
Clearly, the performance (up to random noise) is better than the one of \basic\ at the cost of a higher running-time. 
However, it did not perform as well as the one described above where the decision of whether or not to update a point also depends on the $\delta(X)$- and $\bar{\delta}(X)$-critical $\delta^{*,sn}$-values.

\section{Theoretical Analysis}
\label{sec:analysis}

From our main innovations, namely the non-uniform sampling strategy and the rounding 
procedures ``snapping up'' and ``snapping down'', we already analyzed the snapping
procedures and proved that they enlarge the quality of our estimates. 
The analysis of the non-uniform sampling strategy is much more complicated. One reason 
is that our sampling strategy  strongly interacts with the search heuristic threshold
accepting. That is why we confine ourselves to the analysis of the pure non-uniform
sampling strategy without considering threshold accepting.

In Section \ref{SEC:PROBABILITY}
we prove that sampling in the $d$-dimensional unit cube with respect to the 
probability measure $\pi^d$ instead of $\lambda^d$ leads to superior
discrepancy estimates. (More precisely, we restrict our analysis for technical
reasons to the objective function $\delta$.)
In Section \ref{SUBSEC:RPD} we verify that for $d=1$ sampling with respect to
the probability distribution induced on $\bar{\Gamma}(X)$ by sampling with respect
to $\pi^d$ in $[0,1]^d$ and then rounding to the grid $\bar{\Gamma}(X)$ leads to
better discrepancy estimates than the uniform distribution on $\bar{\Gamma}(X)$.
We comment also on the case $d\ge 2$.
In Section \ref{SUBSEC:NUMBERS} we prove that for random point sets $X$ the probability
of $x\in \bar{\Gamma}(X)$ beeing a critical point is essentially an increasing function of the 
position of its coordinates $x_j$ in the ordered sets $\bar{\Gamma}_j(X)$,
$j=1,\ldots,d$. Recall that critical points yield higher values of the local discrepancy
function $\delta^*$ and include the point that leads to its maximum value. Thus the analysis
in Section \ref{SUBSEC:NUMBERS} serves as another justification of choosing a 
probability measure on the neighborhoods which weights points with larger coordinates
stronger than points with smaller coordinates.

\subsection{Analysis of Random Sampling with Respect to  
$\lambda^d$ and $\pi^d$}
\label{SEC:PROBABILITY}

Here we want to show that sampling in the $d$-dimensional unit cube 
with respect to the non-uniform probability measure $\pi^d$ 
leads to superior results than sampling with respect to the Lebesgue measure
$\lambda^d$.

Before we start with the theoretical analysis, let us give a strong
indication that our non-uniform sampling strategy is much more appropriate
in higher dimension than a uniform sampling strategy.
In~\cite{WF97} Winker and Fang chose in each of the dimensions $d=4,5,6$ in a 
random manner $10$ lattice points sets, cf. also our Section \ref{SEC:EXPERIMENTS}.
They calculated the discrepancy of these sets exactly. If $\eta_d$ 
denotes the average value of the coordinates of the points $y$ with $\delta^*(y) 
= \sup_{z\in [0,1]^d} \delta^*(z)$, we get 
$\eta_4=0.799743$, $\eta_5= 0.840825$, and $\eta_6=0.873523$.
The expected 
coordinate value $\mu_d$ of a point $x$, randomly sampled from $[0,1)^d$
with respect to the measure $\pi^d$, is $\mu_d = d/(d+1)$.
So we get 
$\mu_4 = 0.8$, $\mu_5 =0.8\bar{3}$, and $\mu_6 = 0.857143$.
Note that for using $\lambda^d$ instead of $\pi^d$ the expected coordinate value is
only $0.5$ for all dimensions.

\subsubsection{Random Sampling in the Unit Cube with Respect to 
$\lambda^d$}
\label{SUBSEC:RANDOMDRAWS}

We analyze the setting, where we sample in $[0,1]^d$ with
respect to $\lambda^d$ to maximize 
the objective function $\delta$. A similar analysis for
$\overline{\delta}$ is technically more involved
than the proof of Proposition \ref{FehlerSRA}. 
Furthermore, it leads to a less clear and also
worse result. 
We comment on this at the end of this subsection.

\begin{proposition}
\label{FehlerSRA}
Let $\varepsilon \in (0,1)$, let $n,d \in\N$, and let 
$X=(x^i)_{i=1}^n$ be a sequence in $[0,1)^{d}$. 
Let $x^* = x^*(X)\in [0,1]^d$ satisfy
$\delta(x^*) = \sup_{x\in [0,1]^d}\delta(x)$. Let us assume
that $V_{x^*}\geq \varepsilon$. 
Consider a random point $r\in [0,1]^d$, sampled with respect to the probability 
measure $\lambda^d$. 
If $P^\lambda_\varepsilon = P^\lambda_\varepsilon(X)$
denotes the probability  of the event 
$\{ r\in [0,1]^d \setmid \delta(x^*) - \delta(r) \leq \varepsilon\}$, 
then
\begin{equation}
\label{pepsi}
P^\lambda_\varepsilon \geq \frac{1}{d!} \frac{\varepsilon^d}{V^{d-1}_{x^*}} 
\geq \frac{\varepsilon^d}{d!}\,.
\end{equation}
This lower bound is sharp in the sense that there exist sequences of
point configurations $\{X^{(k)}\}$ such that 
$\lim_{k\to \infty}\,d!\,\varepsilon^{-d} P^\lambda_\varepsilon(X^{(k)})$
converges to 1 as $\varepsilon$ tends to zero.

Let additionally $\epsilon\in (0,1)$ and $R\in\N$. Consider 
random points $r^1,\ldots,r^R \in [0,1]^d$, sampled independently with respect to $\lambda^d$, 
and put $\delta^R := \max^R_{i=1} \delta(r^i)$. If 
\begin{equation}
\label{iteration}
R \geq |\ln(\epsilon)| \Big| \ln \Big( 1-\frac{\varepsilon^d}{d!} \Big) \Big|^{-1}\,,
\end{equation}
then $\delta(x^*) - \delta^R \leq \varepsilon$ with probability at least $1-\epsilon$.
\end{proposition}

Notice, that the case $V_{x^*} < \varepsilon$ left out in 
Proposition \ref{FehlerSRA}
is less important for us,
since our main goal is to find a good lower bound for the star-discrepancy
$\disc(X)$. Indeed, the approximation of $\disc(X)$ up to an admissible error
$\varepsilon$ is a trivial task if $\disc(X) \leq \varepsilon$. 
If $\disc(X) > \varepsilon$, then $V_{x^*} < \varepsilon$ implies 
$\delta(x^*) < \disc(X)$, and the function $\bar{\delta}$
plays the significant role.

\begin{proof}
For $x\leq x^*$ we get
\begin{equation*}
\delta(x) = V_x - \frac{1}{n} \sum^n_{i=1} 1_{[0,x)}(x^i)
\geq \delta(x^*) - (V_{x^*} - V_{x})\,.
\end{equation*}
 Therefore the Lebesgue measure of the set
\begin{equation}
\label{aepsilon}
A_\varepsilon(x^*) := \{x\in [0,1]^d \setmid x\leq x^*, 
V_{x^*}-V_x \leq \varepsilon \}
\end{equation}
is a lower bound for $P^\lambda_\varepsilon$.
Due to Proposition \ref{lemmaLB3}, we have for $d\ge 2$ 
\begin{equation*}
\lambda^d(A_\varepsilon(x^*)) = \frac{1}{d!}\,
\frac{\varepsilon^d}{V^{d-1}_{x^*}}
\sum^\infty_{k=0} b_k(d) \Big( \frac{\varepsilon}{V_{x^*}} \Big)^k
\,,
\end{equation*}
with positive coefficients $b_k(d)$. In particular, we have 
$b_0(d) = 1$. Thus,
\begin{equation}
\label{volab}
\lambda^d(A_\varepsilon(x^*)) \geq \frac{1}{d!}\, 
\frac{\varepsilon^d}{V^{d-1}_{x^*}} 
\geq \frac{\varepsilon^d}{d!}\,,
\end{equation}
and this estimate is obviously also true for $d=1$.
Let us now consider for sufficiently large $k\in\N$ point
configurations 
$X^{(k)}=(x^{(k),i})_{i=1}^n$, where
\begin{equation}
\label{nullmenge}
x^{(k),1}_{1} = \ldots = x^{(k),n}_{1} = k/(k+1) >\varepsilon
\end{equation} 
and
$x^{(k),i}_{j} < k/(k+1) - \varepsilon$ for all $i\in [n]$, $j >1$. 
Then obviously $x^*(X^{(k)}) = (k/(k+1), 1,\ldots, 1)$, and it is easy
to see that $P^\lambda_\varepsilon(X^{(k)}) = \lambda^d(\A(x^*))$. 
From Proposition \ref{lemmaLB3} we get
\begin{equation*}
\lambda^d(\A(x^*)) = 
\bigg( \frac{k+1}{k} \bigg)^{d-1}
\frac{\varepsilon^d}{d!} \bigg(
1 + O \bigg( \frac{k+1}{k}\,\varepsilon \bigg) \bigg)\,.
\end{equation*}
This proves that estimate (\ref{pepsi}) is sharp. Notice that we
assumed (\ref{nullmenge}) only for simplicity. Since $\delta(x^*(X))$ is
continuous in $X$, we can find for fixed $k$ an open set of point
configurations doing essentially the same job as $X^{(k)}$.

Assume now
$\delta(x^*) - \delta^R > \varepsilon$, i.e., 
$\delta(x^*) - \delta(r^i) > \varepsilon$ for all $i\leq R$. The probability
of this event is at maximum $(1 - \varepsilon^d/d!)^R$. This probability is 
bounded from above by $\epsilon$ if $R$ satisfies 
(\ref{iteration}).
\end{proof}

For $d\geq 1$ and $\varepsilon \leq 1/2$ we have
$\big| \ln (1- \varepsilon^d/d!) \big|^{-1}
\sim \,d!\, \varepsilon^{-d}\,.$
In this case we can only assure that 
$\delta^R$ is an $\varepsilon$-approximation of 
$\sup_{x\in [0,1]^d}\delta(x)$ with a certain probability if 
the number $R$ of randomly sampled points is super-exponential in $d$.

Let us end this section with some comments on the setting where 
we are only interested in maximizing
$\bar{\delta}$. If for given $\varepsilon > 0$, $X\in [0,1)^{nd}$
the maximum of $\bar{\delta}$ is achieved in $\bar{x} = \bar{x}(X)
\in [0,1]^d$, and if we want to know the probability of the 
event $\{ r\in [0,1]^d \setmid \bar{\delta}(\bar{x}) 
- \bar{\delta}(r) \leq \varepsilon \}$, there seems to be
no alternative to estimating $\lambda^d(U(\bar{x}))$, where
\begin{equation*}
U(\bar{x}) := \{r\in [0,1]^d \setmid \bar{x} \leq r\,,
\,V_r - V_{\bar{x}} \leq \varepsilon \}\,.
\end{equation*}
It is easy to see that $\lambda^d(U(\bar{x}(X)))$ approaches  
zero if one of the coordinates of $\bar{x}$ tends to 
1 -- regardless of $\varepsilon$ and $V_{\bar{x}}$. We omit a tedious error 
analysis to cover the $\bar{\delta}$-setting. 
 
\subsubsection{Random Sampling in the Unit Cube with Respect to $\pi^d$}
\label{SUBSEC:POLY}

Similarly as in the preceding section, we analyze here the 
setting where, in order to maximize $\delta$, we sample in $[0,1]^d$ with respect to $\pi^d$.

\begin{proposition}
\label{FehlerCPM}
Let $\varepsilon,d,n,X=(x^i)_{i=1}^n$ and $x^*$ as in Proposition~\ref{FehlerSRA}.
Again assume $V_{x^*} \geq \varepsilon$. 
Consider a random point $r\in[0,1]^d$, sampled with respect to the 
probability measure $\pi^d$. 
If $P^\pi_\varepsilon = P^\pi_\varepsilon(X)$
denotes the probability  of the event 
$\{ r\in [0,1]^d \setmid \delta(x^*) - \delta(r) \leq \varepsilon\}$, 
then $P^\pi_\varepsilon \geq \varepsilon^d$.
This lower bound is sharp, since there exists a point configuration
$X$ such that $P^\pi_\varepsilon(X) = \varepsilon^d$.

Let additionally $\epsilon\in (0,1)$ and $R\in\N$. Consider 
random points $r^1,\ldots,r^R \in [0,1]^d$, sampled independently with respect to $\pi^d$, 
and put $\delta^R := \max^R_{i=1} \delta(r^i)$. If 
\begin{equation}
\label{iteration2}
R \geq |\ln(\epsilon)| | \ln( 1- \varepsilon^d)|^{-1}\,,
\end{equation}
then $\delta(x^*) - \delta^R \leq \varepsilon$ with probability at least $1-\epsilon$.
\end{proposition}

\begin{proof}
Clearly $P^\pi_\varepsilon \geq \pi^d(\A(x^*))$, where $\A(x^*)$ is defined
as in (\ref{aepsilon}). Due to Proposition \ref{B2} we have 
$\pi^d(\A(x^*)) \geq \varepsilon^d$. Let us now consider the point
configuration $X$, where $x^{1}_1 = \ldots = x^n_1 = \varepsilon$ 
and $x^i_j < \varepsilon$ for all $i\in [n]$, $j>1$. 
Furthermore, at least an $\varepsilon^{-1}$-fraction of the points
should be equal to $(\varepsilon, 0,\ldots,0)$. Then obviously
$x^*(X) = (\varepsilon,1,\ldots,1)$ and $P^\lambda_\varepsilon(X)
= \pi^d(\A(x^*)) = \varepsilon^d$. 
 
Let us now assume that $\delta(x^*) - \delta^R > \varepsilon$, i.e., 
$\delta(x^*) - \delta(r_i) > \varepsilon$ for all $i\leq R$. This happens
with probability not larger than $(1-\varepsilon^d)^R$. Therefore we have 
$(1-\varepsilon^d)^R \leq \epsilon$ if $R$ satisfies (\ref{iteration2}).
\end{proof}

If $d\geq 1$ and $\varepsilon \leq 1/2$, then
$|\ln(1-\varepsilon^d)|^{-1} \sim \varepsilon^{-d}$. Here the number
of iterations $R$ ensuring with a certain probability that $\delta^R$ 
is an $\varepsilon$-approximation of 
$\sup\{\delta(x) \setmid x\in [0,1]^d\}$ is still exponential in $d$,
but at least not super-exponential as in the previous
section.

Altogether we see that a simple sampling algorithm relying on the probabilistic
measure $\pi^d$ rather than on $\lambda^d$ is more likely to find larger values of
$\delta$.

\subsection{Analysis of Rounding to the Coordinate Grid}
\label{SUBSEC:RPD}
As described in Sections \ref{SUBSEC:NEIGHBORS} and \ref{subsec:starting}, our non-uniform
sampling strategy on the grids $\bar{\Gamma}(X)$ and $\Gamma(X)$ for the objective
functions $\delta$ and $\bar{\delta}$ consists of sampling in $[0,1]^d$ with 
respect to $\pi^d$ and then rounding the sampled point $y$ up and down to
grid points $y^+$ and $y^-$, respectively. This induces discrete probability
distributions $w_u = (w_u(z))_{z\in\bar{\Gamma}(X)}$ and 
$w_l = (w_l(z))_{z\in\Gamma(X)}$ on $\bar{\Gamma}(X)$ and $\Gamma(X)$, respectively.
If we use additionally the rounding procedures ``snapping up'' and ``snapping down'',
as described in Section \ref{SUBSEC:SNAPPING}, this will lead to modified probabilistic
distributions $w^{sn}_u= (w^{sn}_u(z))_{z\in\bar{\Gamma}(X)}$ and 
$w^{sn}_l = (w^{sn}_l(z))_{z\in\Gamma(X)}$ on $\bar{\Gamma}(X)$ and $\Gamma(X)$, respectively.
In dimension $d=1$ the probability distributions $w_u$ and $w^{sn}_u$ as well as 
$w_l$ and $w_l^{sn}$ are equal, since every test box is a critical one.
Essentially we prove in the next section that in the one-dimensional case 
sampling with respect to the probability distributions
$w_u = w_u^{sn}$ [$w_l = w_l^{sn}$] leads to larger values of $\delta$ [$\bar{\delta}$]
than sampling with respect to the uniform distribution on $\bar{\Gamma}(X)$ [$\Gamma(X)$].

\subsubsection{Analysis of the 1-Dimensional Situation}

Recall that in the 1-dimensional case $\pi =\pi^1$ coincides with 
$\lambda =\lambda^1$. 

To analyze the 1-dimensional situation, let $X:=(x^i)_{i=1}^n$ be the 
given point configuration in $[0,1)$. 
Without loss of generality we assume that $0 \leq x^1 < \cdots < x^n <1$. 
Since $\delta^*(1) = 0$ we do not need to consider the whole grid $\bar{\Gamma}(X)$ but can restrict
ourselves to the set $\Gamma(X) = \{x^1,\ldots,x^n\}$. 
For the same reason, let us set $y^+:=x^1$ if $y >x^n$ (recall that, following the description given in Section~\ref{SUBSEC:NEIGHBORS}, we set $y^-:=x^n$ for $y<x^1$ anyhow).

As discussed above, we take points randomly from $\Gamma(X)$, but instead of 
using equal probability weights on $\Gamma(X)$, we 
use the probability distributions $w_u = w_u^{sn}$ and $w_l=w_l^{sn}$ 
on $\Gamma(X)$ to maximize our objective functions $\delta$ 
and $\bar{\delta}$, respectively.   
If we put $x^0 := x^n - 1$ and $x^{n+1} := x^1 + 1$, then the corresponding 
probability weights for $\delta$ and $\bar{\delta}$ are given
by $w_l(x^i) := x^i - x^{i-1}$ and $w_u(x^i) := x^{i+1} - x^i$, respectively.

In the next lemma we will prove the following statements rigorously:
If one wants to sample a point
$\tau\in\Gamma(X)$ with $\delta(\tau)$ as large as possible or 
if one wants to enlarge the chances to sample the point $\tau$ where 
$\delta$ takes its maximum, its preferable to use the weights $w_l$
instead of the equal weights $1/n$ on $\Gamma(X)$. 
Similarly, it is preferable to employ the weights $w_u(x^i)$, $i=1,\ldots,n$,
instead of equal weights if one wants to increase the expectation of 
$\bar{\delta}$ or the chances of sampling the maximum of 
$\bar{\delta}$. 

\begin{lemma}
\label{1dHeuristik}
Let $d=1$ and $\tau$, $\bar{\tau}\in \Gamma(X)$ with 
$\delta(\tau) = \sup_{z\in [0,1]} \delta(z)$ and 
$\bar{\delta}(\bar{\tau}) = \sup_{z\in [0,1]} \bar{\delta}(z)$.
Then we have
$w_l(\tau) \geq 1/n$ and
$w_u(\bar{\tau}) \geq 1/n$.

Furthermore, let $\Expec$, $\Expec_l$, and $\Expec_u$ denote the 
expectations with respect to the uniform 
weights $\{1/n\}$, the weights $\{w_l(x^i)\}$, and the weights 
$\{w_u(x^i)\}$ on the 
probability space $\Gamma(X)$, respectively. Then
$\Expec_l(\delta) \geq \Expec(\delta)$
and
$\Expec_u(\bar{\delta}) \geq \Expec(\bar{\delta})$.
\end{lemma}

\begin{proof}
Let $\nu\in [n]$ with $\tau = x^\nu$.
Assume first $w_l(x^\nu) < 1/n$, i.e., $x^\nu - x^{\nu-1} < 1/n$.
If $\nu>1$, then
\begin{equation*}
\delta(x^{\nu-1}) = x^{\nu-1} - \frac{\nu-2}{n} > 
x^\nu - \frac{\nu-1}{n} = \delta(x^\nu)\,.
\end{equation*}
If $\nu=1$, then however
\begin{equation*}
\delta(x^n) = x^n - \frac{n-1}{n} = x^0 + \frac{1}{n} > x^1 = \delta(x^\nu)\,.
\end{equation*}
So both cases result in a contradiction.

We prove now $\Expec_l(\delta) \geq \Expec(\delta)$ by induction over the 
cardinality $n$ of $X$.
For $n=1$ we trivially have $\Expec_l(\delta) = x^1 = \Expec(\delta)$.

Therefore, let the statement be true for $n$ and consider an ordered set
$\Gamma(X) := \{x^1,\ldots, x^{n+1}\}$. Let $\delta$ achieve its maximum
in $x^{\nu} \in \Gamma(X)$. 
We already proved that  
$w_l(x^\nu) = x^\nu - x^{\nu-1} \geq 1/(n+1)$ holds. With the notation
\begin{equation*}
\tilde{x}^i := x^i \hspace{2ex}\text{if $1\leq i <\nu$,}\hspace{3ex}
\tilde{x}^i := x^{i+1} - \frac{1}{n+1} \hspace{2ex}\text{if $i\geq \nu$,}
\end{equation*}
and 
\begin{equation*}
\hat{x}^i := \frac{n+1}{n}\,\tilde{x}^i\,,\hspace{2ex}i=1,\ldots,n\,,
\hspace{2ex}\text{and}\hspace{2ex}
\hat{x}^0 := \hat{x}^n -1\,,
\end{equation*}
we get
\begin{equation*}
\begin{split}
\Expec_l(\delta) &=
\sum^{n+1}_{i=1} w_l(x^i)\delta(x^i)
= \frac{\delta(x^\nu)}{n+1} 
+ \Big( w_l(x^\nu) - \frac{1}{n+1} \Big) \delta(x^\nu) 
+  \sum^{n+1}_{i=1\atop i\neq \nu} w_l(x^i)\delta(x^i)\\
&\geq\, \frac{\delta(x^\nu)}{n+1} 
+ \Big( w_l(x^{\nu+1}) + w_l(x^\nu) - \frac{1}{n+1} \Big) \delta(x^{\nu+1}) 
+  \sum^{n+1}_{i=1\atop i\notin \{\nu, \nu +1\}} w_l(x^i)\delta(x^i)\\
&=\, \frac{\delta(x^\nu)}{n+1}+ \Big( \tilde{x}^{1} - \tilde{x}^n + \frac{n}{n+1} 
\Big)\,\tilde{x}^1 + \sum^n_{i=2} (\tilde{x}^i - \tilde{x}^{i-1}) 
\Big( \tilde{x}^i - \frac{i-1}{n+1} \Big)\\
&=\, \frac{\delta(x^\nu)}{n+1}+ 
\Big( \frac{n}{n+1} \Big)^2 \sum^n_{i=1} (\hat{x}^i - \hat{x}^{i-1}) 
\Big( \hat{x}^i - \frac{i-1}{n} \Big)\,.
\end{split}
\end{equation*}
On the other hand we have
\begin{equation*}
\begin{split}
\Expec(\delta) &= 
\sum^{n+1}_{i=1} \frac{1}{n+1} \delta(x^i)
=\, \frac{\delta(x^\nu)}{n+1} + \sum^n_{i=1}
\frac{1}{n+1} \Big( \tilde{x}^i - \frac{i-1}{n+1} \Big)\\
&=\, \frac{\delta(x^\nu)}{n+1} + \Big( \frac{n}{n+1} \Big)^2 \sum^n_{i=1}
\frac{1}{n} \Big( \hat{x}^i - \frac{i-1}{n} \Big)\,.
\end{split}
\end{equation*}
These calculations and our induction hypothesis, applied to  
$\{\hat{x}^1,\ldots, \hat{x}^n\}$, lead to
$\Expec_l(\delta) \geq \Expec(\delta)$. 
For $\mu \in [n]$ with $\bar{\tau}=x^\mu$
the inequalities
$w_u(x^\mu)\geq 1/n$ and 
$\Expec_u(\bar{\delta}) \geq \Expec(\bar{\delta})$
can be proved in a similar manner.
\end{proof}

\subsubsection{Analysis of Higher Dimensional Situations $d \geq 2$}
\label{SUBSUBSEC:WEIGHTS}
If we turn to the situation in dimension $d\geq 2$, we first face a technical
difference: For many sequences $X$ the supremum
$\sup_{x\in [0,1]^d} \delta(x)$ will be achieved by some $x\in \bar{\Gamma}(X)$ with
$x_j = 1\notin \{x^{1}_j,\ldots, x^n_j\}$ for at least one $j\in [d]$. Therefore,
we typically have to consider the whole grid $\bar{\Gamma}(X)$ if we want to 
maximize $\delta$. 

Let us now have a look at the weight functions $w_l$, $w_u$ 
induced by our algorithms in dimension $d\geq 2$. Actually,
here we only consider the simpler Lebesgue measure $\lambda^d$, but it is obvious 
how to modify the definitions
and arguments below to cover the $\pi^d$-setting.

For all $j\in [d]$ let $\nu_j := |\bar{\Gamma}_j(X)|$. As in Section~\ref{SUBSEC:WF} let
$\phi_j: [ \nu_j ]\to \bar{\Gamma}_j(X)$ be the ordering of the 
set $\bar{\Gamma}_j(X)$. Set $\phi_j(0):=0$.
For $y=(\phi_1(i_1), \ldots, \phi_d(i_d))\in \bar{\Gamma}(X)$, 
$i_1,\ldots,i_d \in [\nu_j]$, we define the weight $w_l(\phi)$ by
\begin{equation*}
w_l(y) := \prod^d_{j=1} \big( \phi_j(i_j) - \phi_j(i_j-1) \big)\,.
\end{equation*} 
For the definition of the weights 
$w_u(y)$ let $\tilde{\phi}_j(i)=\phi_j(i)$ for 
$i \in [\nu_j-1]$ and let $\tilde{\phi}_j(\nu_j):=\phi_j(1)+1$. 
For $\tilde{y}=(\tilde{\phi}_1(i_1), \ldots, \tilde{\phi}_d(i_d)) \in \Gamma(X)$,
$i_1,\ldots, i_d \in [\nu_j -1]$, let
\begin{equation*}
w_u(y) := \prod^d_{j=1} \big( \tilde{\phi}_j(i_j+1) - \tilde{\phi}_j(i_j) \big)\,.
\end{equation*}

Let $y \in\bar{\Gamma}(X)$ and $\tilde{y}\in \Gamma (X)$. Then the weights 
$w_l(y)$ and $w_u(\tilde{y})$ are obviously the probabilities
that after sampling a point $z$ in $[0,1]^d$ with respect to $\lambda^d$,
we end up with $z^{+} = y$ and $z^{-} = \tilde{y}$, respectively.

\begin{figure}
\center 
\psfrag{1}{$x^1$}
\psfrag{2}{$x^2$} 
\psfrag{3}{$x^3$}
\psfrag{4}{$x^4$}
\psfrag{5}{$x^5$}
\psfrag{6}{$\gamma$}

{\epsfig{file=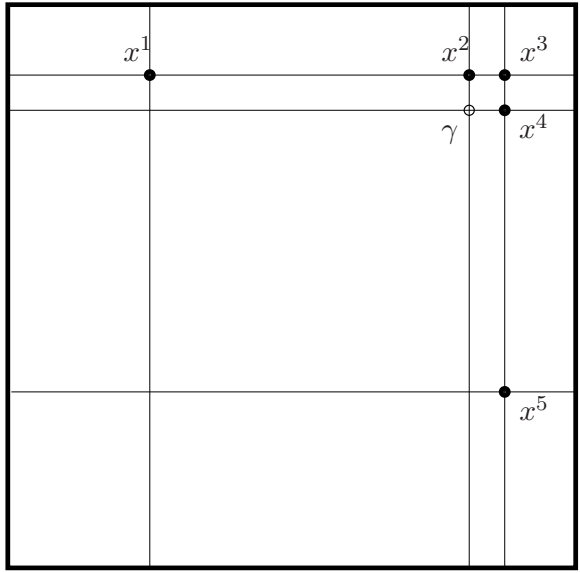,width=.6\textwidth}}
\caption{\label{fig_lower} $\delta$ takes its maximum in $x^3$, but
  $w_l(x^3) < 1/36$}
\end{figure}

\begin{figure}
\center 
\psfrag{1}{$x^1$}
\psfrag{2}{$x^2$} 
\psfrag{3}{$x^3$}
\psfrag{4}{$x^4$}
\psfrag{5}{$x^5$}
\psfrag{6}{$\gamma$}
{\epsfig{file=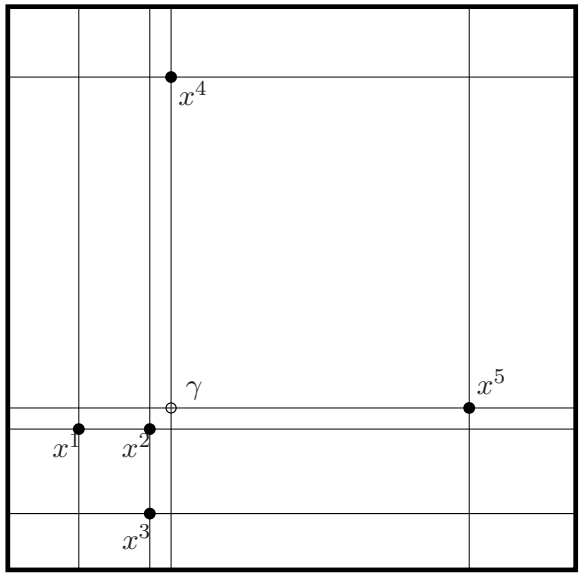,width=.6\textwidth}}
\caption{\label{fig_upper} $\bar{\delta}$ takes its maximum in
  $x^2$, but $w_u(x^2) < 1/36$}
\end{figure} 

The simple examples in Figures \ref{fig_lower} and \ref{fig_upper} with
$d=2$ and $n=5$
illustrate that we cannot
prove the extension of the weight inequalities from
Lemma \ref{1dHeuristik} in $d\geq 2$.
More precisely: 
If $\tau\in \bar{\Gamma}(X)$ and $\bar{\tau}\in \Gamma(X)$ are points with
$\delta(\tau) = \sup_{y\in [0,1]^d} \delta(y)$ and 
$\bar{\delta}(\bar{\tau}) = \sup_{y\in [0,1]^d} \bar{\delta}(y)$,
then in general we do \emph{not} have
\begin{equation}
\label{giltnicht}
w_l(\tau) \geq \prod^d_{j=1}\frac{1}{\nu_j }
\text{ and }
w_u(\bar{\tau}) \geq \prod^d_{j=1}\frac{1}{\nu_j-1}\,,
\end{equation} 
even $w_l(\tau)$ or $w_u(\bar{\tau}) \geq (n+1)^{-d}$ need not to
be satisfied.
Notice that for fixed $d$ and $n$ the sets of counterexamples to each of the 
inequalities in (\ref{giltnicht}) have strictly positive
measure. Indeed, the suprema of $\delta$ and $\bar{\delta}$ 
are continuous in $X$, and the
weights $w_l$, $w_u$ are continuous in $X$ as long as for each $j\in [d]$
all coordinates $x^1_j,\ldots, x^n_j$ are distinct. But it
is, e.g., easy to rearrange the counterexamples in Figure \ref{fig_lower} and
\ref{fig_upper} slightly such that these constraints are satisfied.

But notice also that in both figures the grid point $\gamma$ provides
a good approximation of the actual maximum of $\delta$ and
$\bar{\delta}$, and $\gamma$ has a rather large weight
$w_l$ and $w_u$, respectively.
Furthermore, the point sets in Figure \ref{fig_lower} and
\ref{fig_upper} are quite artificial and also far away from being
well-distributed. 
For random or low-discrepancy sets these probability weights can still
be very useful -- this is also confirmed by our numerical experiments,
see Section~\ref{SEC:EXPERIMENTS}. 
Moreover, the invalidity of (\ref{giltnicht}) does not 
necessarily mean that the expectations $\Expec(\delta)$ and 
$\Expec(\bar{\delta})$ are larger than $\Expec_l(\delta)$
and $\Expec_u(\bar{\delta})$, respectively.

Actually, if we use additionally the procedures ``snapping up'' and ``snapping
down'' and look at the induced probability distribution  $w_l^{sn}$ and $w_u^{sn}$,
respectively, then $x^3$ has the maximum weights $w_l^{sn}$ of all critical 
points in Figure \ref{fig_lower}, and $x^2$ has the maximum weight $w^{sn}_u$ of all
critical points in Figure \ref{fig_upper}. Thus adding the snapping procedures can 
change the situation decisively.

Additional theoretical results in this direction
would be interesting. 

\subsection{The Chances of a Grid Point Beeing a Critical Point}
\label{SUBSEC:NUMBERS}

If $X=(x^1,\ldots,x^n)$ is a sequence in $[0,1)^d$ which has been chosen randomly with 
respect to the Lebesgue measure and if $x\in\bar{\Gamma}(X)$, then the larger the components of $x$,
the higher is the probability of $x$ being a $\delta(X)$-critical
point. 
The same holds for $\tilde{x}\in\Gamma(X)$ and 
$\bar{\delta}(X)$, respectively. 

\begin{proposition}
\label{Deltavert}
Consider $[0,1)^{nd}$ as a probability space endowed with the 
probability measure $\lambda^{nd}$. Let $\iota := (i_1,\ldots, i_d)
\in [n+1]^d$.
If $k$ indices $i_{\nu(1)},\ldots,i_{\nu(k)}$ of the $i_j$, $j=1, \ldots,d$, 
are at most $n$ and the remaining 
$d-k$ of them are equal to $n+1$, then for uniformly distributed 
random variable $X$ in $[0,1)^{nd}$
the multi-index
$\iota$ is $\delta(X)$-critical with probability
\begin{equation*}
\bigg( \frac{(n-k)!}{n!} \bigg)^{k-1} 
\prod^k_{j=1} \bigg( \prod^{k-1}_{\ell =1} \max\{i_{\nu(j)} -\ell,0\} \bigg)\,.
\end{equation*}
\end{proposition}

\begin{proof}
Let $\Phi=(\phi_1,\ldots,\phi_d)$ 
be as in Section~\ref{SUBSEC:SNAPPING}.
Since the event  
that for all coordinates $j \in [d]$ we have $|\bar{\Gamma}_j(X)|=n+1$
holds with probability $1$, we restrict ourselves to this 
situation. 

Without loss of generality, we may assume that $i_1,\ldots,i_k \leq n$ 
and $i_{k+1} = \ldots = i_d = n+1$.  Obviously, 
$S_{k+1}(\Phi(\iota)), \ldots, S_d(\Phi(\iota))$ are $\delta$-critical
surfaces, since $\phi_{k+1}(n+1) =\ldots = \phi_d(n+1) =1$.
For $i=1,\ldots,k$ let $\sigma_i = \sigma_i(X):[n]\to [n]$ 
be the permutation with
\begin{equation*}
x^{\sigma_i(1)}_i < x^{\sigma_i(2)}_i < \cdots < x^{\sigma_i(n)}_i 
< 1\,.
\end{equation*}
Clearly, $\Phi(\iota)_j = \phi_j(i_j) = x^{\sigma_j(i_j)}_j$ for all
$j\in [k]$.
Since $S_i(x)\cap S_j(x) = \emptyset$ for all $i\neq j$ and all 
$x\in [0,1]^d$, the surfaces $S_{1}(\Phi(\iota)), \ldots, S_k(\Phi(\iota))$
can only be $\delta(X)$-critical if 
$|\{ \sigma_1(i_1),\ldots, \sigma_k(i_k)\}| = k$.
More precisely, $\iota$ is a $\delta(X)$-critical multi-index if and only
if the condition
\begin{equation*}
\forall j \in [k]\, \forall l\in [k]\setminus \{j\}\,:\, 
x^{\sigma_j(i_j)}_l < \Phi(\iota)_l = x^{\sigma_l(i_l)}_l
\end{equation*}
holds. This is equivalent to the $k$ conditions  
\begin{equation}
\label{bedingungen}
\begin{array}{*{3}{c@{}}c@{}cc}
\sigma^{-1}_1(\sigma_2(i_2))\,,\,\, & \sigma^{-1}_1(\sigma_3(i_3))\,,\, &
\ldots\,, & \sigma^{-1}_1(\sigma_k(i_k))& < & i_1\,,\\
\sigma^{-1}_2(\sigma_1(i_1))\,,\,\, & \sigma^{-1}_2(\sigma_3(i_3))\,,\, &
\ldots\,, & \sigma^{-1}_2(\sigma_k(i_k)) & < & i_2\,,\\
\vdots & \vdots & & \vdots & \vdots& \vdots\\
\sigma^{-1}_k(\sigma_1(i_1))\,,\,\, & \sigma^{-1}_k(\sigma_2(i_2))\,,\, &
\ldots\,,\, & \sigma^{-1}_k(\sigma_{k-1}(i_{k-1}))\,\, & < & i_{k}\,.\\
\end{array}
\end{equation}
Since all the components $x^i_j$, $i\in [n]$, $j\in [d]$,
of $X$ are independent random variables, we have that for a fixed index $\nu\in [d]$
each permutation $\tau: [n]\to[n]$ is equally likely to 
fulfill $\sigma_\nu(X) = \tau$. Thus, the probability of $\iota$ 
being a $\delta(X)$-critical 
index is just the number of $k$-tuples $(\sigma_1,\ldots,\sigma_k)$ 
of permutations
fulfilling (\ref{bedingungen}), divided by $(n!)^k$.

For given pairwise distinct values $\sigma_1(i_1), \ldots, \sigma_k(i_k)$ the 
$j$-th condition of (\ref{bedingungen}) is satisfied by 
$((i_j-1)\ldots (i_j-(k-1)))(n-k)!$ permutations $\sigma_j$. 
Since all $k$ conditions in (\ref{bedingungen}) can be solved independently
of each other, it is now easy to deduce the statement of the Proposition.
\end{proof}

To state the corresponding proposition for $\bar{\delta}$, we have to 
introduce \emph{Stirling numbers of second kind} $S(d,k)$.
For $k\in \N$, $k\leq d$ let $S(d,k)$ denote the number of partitions of
$[d]$ into $k$ non-empty subsets. A closed formula for $S(d,k)$ is
\begin{equation*}
S(d,k) = \sum^k_{j=0} \frac{(-1)^j\, (k-j)^d}{j!\, (k-j)!}\,.
\end{equation*}
This formula and other useful identities can, e.g., be found in~\cite{Rio}.

\begin{proposition}
\label{Deltaquervert}
Let $X$ be a uniformly distributed random variable in $[0,1)^{nd}$.
Let $\iota = (i_1,\ldots,i_d) \in [n]^d$. Then $\iota$ is a 
$\bar{\delta}(X)$-critical multi-index with probability
\begin{equation*}
\sum^d_{k=1} S(d,k) 
\bigg( \frac{(n-k)!}{n!} \bigg)^{d-1} 
\prod^d_{j=1} \bigg( \prod^{k-1}_{\nu =1} (i_j -\nu) \bigg)\,.
\end{equation*}
\end{proposition}

\begin{proof}
We just need to consider the case where the almost-sure
event 
$|\Gamma_j(X)| = n$ for all $j\in [d]$ holds. 
For $j=1,\ldots,d$ let $\sigma_j := \sigma_j(X):[n]\to [n]$
be the permutation with
\begin{equation*}
x^{\sigma_j(1)}_j < x^{\sigma_j(2)}_j < \cdots < x^{\sigma_j(n)}_j 
< 1\,.
\end{equation*}
Then $\Phi(\iota)_j = \phi_j(i_j) = x^{\sigma_j(i_j)}_j$ for all
$j\in [d]$.
It is easy to see that the surface $\bar{S}_{j}(\Phi(\iota))$
is $\bar{\delta}(X)$-critical if 
and only
if the condition
\begin{equation*}
\forall j \in [d]\, \forall l\in [d]\setminus \{j\}\,:\, 
x^{\sigma_j(i_j)}_l \leq \Phi(\iota)_l = x^{\sigma_l(i_l)}_l 
\end{equation*}
is satisfied. This can be rewritten as  
\begin{equation}
\label{bedingungenquer}
\begin{array}{*{3}{c@{}}c@{}cc}
\sigma^{-1}_1(\sigma_2(i_2))\,,\,\, & \sigma^{-1}_1(\sigma_3(i_3))\,,\, &
\ldots\,, & \sigma^{-1}_1(\sigma_d(i_d))& \leq & i_1\,,\\
\sigma^{-1}_2(\sigma_1(i_1))\,,\,\, & \sigma^{-1}_2(\sigma_3(i_3))\,,\, &
\ldots\,, & \sigma^{-1}_2(\sigma_d(i_d)) & \leq & i_2\,,\\
\vdots & \vdots & & \vdots & \vdots& \vdots\\
\sigma^{-1}_d(\sigma_1(i_1))\,,\,\, & \sigma^{-1}_d(\sigma_2(i_2))\,,\, &
\ldots\,,\, & \sigma^{-1}_d(\sigma_{d-1}(i_{d-1}))\,\, & \leq & i_{d}\,.\\
\end{array}
\end{equation}
If $|\{\sigma_1(i_1),\ldots, \sigma_d(i_d)\}| = k$, then
there exist
\begin{equation*}
S(d,k)\, n!\,\big( (n-k)! \big)^{d-1}\, \prod^d_{j=1}\prod^{k-1}_{\nu = 1}
(i_j - \nu)
\end{equation*}
permutations satisfying (\ref{bedingungenquer}). With this observation and the 
fact that  
all components $x^i_j$, $i\in [n]$, $j\in [d]$,
of $X$ are stochastically independent, it is now easy to deduce the statement 
of Proposition \ref{Deltaquervert}. 
\end{proof}

\section{Experimental Results}\label{sec:experiments}\label{SEC:EXPERIMENTS}
We now present the experimental evaluations of the algorithms. 
We will compare our basic and improved algorithms, \basic\ and \improved,
against the algorithm of Winker and Fang~\cite{WF97},
and also give a brief comparison against the genetic algorithm of Shah~\cite{Sha10} and 
the integer programming-based algorithm of Thi\'emard~\cite{Thi01b}.

\subsection{Experimental Setup}

We divide our experiments into a thorough comparison against the algorithm of
Winker and Fang~\cite{WF97}, given in Section~\ref{subsec:expwf},
and more brief comparisons against the algorithms of Shah~\cite{Sha10}
and Thi\'emard~\cite{Thi01b}, in Sections~\ref{subsec:expshah} and~\ref{subsec:expthie}, respectively.
The algorithms \basic\ and \improved,  as well as the algorithm of Winker and
Fang~\cite{WF97}, were implemented by the authors in the C programming
language, based on the code used in~\cite{Win07}.  All implementations were
done with equal care.  In the case of Winker and Fang~\cite{WF97}, while we
did have access to the original Fortran source code (thanks to P.~Winker), due to lack
of compatible libraries we could not use it, and were forced to do a
re-implementation. 

For the integer programming-based algorithm of Thi\'emard~\cite{Thi01b}, 
E.~Thi\'emard has kindly provided us use of the source code.  This source
code was modified only as far as necessary for compatibility with newer
software versions -- specifically, we use version 11 of the CPLEX integer
programming package, while the code of Thi\'emard was written for an older
version. Finally, Shah has provided us with the application used in the
experiments of~\cite{Sha10}, but as this application is hard-coded to
use certain types of point sets only, we restrict ourselves to comparing 
with the experimental data published in~\cite{Sha10}. 
Random numbers were generated using the Gnu C library pseudorandom number
generator. 

The instances used in the experiments are described in
Section~\ref{subsec:expinstances}.  
For some instances, we are able to compute the exact discrepancy values
either using an implementation of the algorithm of Dobkin et al.~\cite{DEM96},
available from the third author's homepage\footnote{Found at \url{http://www.mpi-inf.mpg.de/~wahl/}.},
or via the integer programming-based algorithm of Thi\'emard~\cite{Thi01b}. 
These algorithms both have far better time dependency than that of Bundschuh and Zhu~\cite{BZ93},
allowing us to report exact data for larger instances than previously done.
For those instances where this is too costly, we report instead the largest
discrepancy value found by any algorithm in any trial; these imprecise values
are marked by a star.
Note that this includes some trials with other (more time-consuming) parameter
settings than those of our published experiments; thus sometimes, none of the
reported algorithms are able to match the approximate max value.

As parameter settings for the neighbourhood for \basic, we use $\ell=\lfloor \frac{n}{8} \rfloor$ if $n \geq 100$, and $\ell =\lfloor \frac{n}{4} \rfloor$ otherwise, and $mc=2$ throughout.
These settings showed reasonable performance in our experiments and in \cite{Win07}.
For \improved, these parameters are handled by the scaling described in Section~\ref{subsec:shrinking}.

Throughout, for our algorithms and for the Winker and Fang algorithm, 
we estimate the expected outcome of running 10 independent trials of 
100,000 iterations each and returning the largest discrepancy value found, and
call this the \emph{best-of-10 value}.
The estimation is computed from a basis of 100 independent trials, as suggested by Johnson~\cite{Johnson},
which strongly decreases irregularities due to randomness compared to
the method of taking 10 independent best-of-10 values and averaging these.
The comparisons are based on a fix number of iterations, rather than equal
running times, as the point of this paper is to compare the strengths of the 
involved concepts and ideas, rather than implementation tweaks. 
For this purpose, using a re-implementation rather than the original algorithm of 
Winker and Fang~\cite{WF97} has the advantage that all algorithms compared use the same code
base, compiler, and libraries, including the choice of pseudo-random number generator.
This further removes differences that are not interesting to us.

\subsection{Instances} \label{subsec:expinstances}

Our point sets are of four types: Halton sequences~\cite{Hal60}, Faure sequences~\cite{Faure},
Sobol' point sets~\cite{Sobol}, and so-called Good Lattice Points (GLP), described below.
The Halton sequences and GLPs were generated by programs written by the authors,
the Faure sequences by a program of John Burkardt~\cite{faureurl}, and the Sobol' sequences
using the data and code of Stephen Joe and Frances Kuo~\cite{JK08, sobolurl}.

Winker and Fang tested their algorithm for several point sets 
in dimension $d = 4, 5, \ldots, 11$, which were constructed in the following manner:
Let $(n, h_1,\ldots,h_d) \in \N^{d+1}$ with 
$0< h_1 < h_2 < \cdots < h_d < n$, where at least one $h_i$ is
relatively prime with $n$, i.e., their greatest common divisor 
is one. Then the points $x^1,\ldots,x^n \in [0,1)^d$ are given
by 
\begin{equation*}
x^i_j := \Big\{ \frac{2ih_j -1}{2n} \Big\}\,,\hspace{2ex}
i\in [n], j\in [d]\,,
\end{equation*}
where $\{x\}$ denotes the fractional part of $x\in\R$, i.e.,
$\{x\} = x - \lfloor x \rfloor$.  
Winker and Fang call $\{x^1,\ldots,x^n\}$  a \emph{good lattice point} (GLP) 
set\footnote{Other authors call it GLP set if it additionally exhibits a small
discrepancy.} of the generating vector $(n, h_1,\ldots,h_d)$.
It is known that for any $d\ge 2$ and $n \ge 2$ there exists a generating vector such that the
corresponding GLP set exhibits asymptotically a discrepancy of $O(\log(n)^d/n)$,
where the implicit constant of the big-O-notation depends solely on $d$, see, e.g.,
\cite[Sect.~5.2]{Nie92}.

Winker and Fang considered two series of examples to test their
algorithm: First, they (randomly) generated in each dimension 
$d=4,5,6$
ten GLP $n$-point sets, where $n\in\{50,51,\ldots,500\}$ for 
$d=4$, $n\in\{50, 51,\ldots, 250\}$ for $d=5$ and 
$n\in\{25,26,\ldots,100\}$ for $d=6$. 
For each GLP set the exact discrepancy was calculated with an
implementation of the algorithm of Bundschuh and Zhu \cite{BZ93}.

Secondly, they considered six GLP sets in dimension $d=6,7,\ldots, 11$
with cardinality between $2129$ and $4661$ points and performed
$20$ trials with $200,000$ iterations for each of the six sets.
Solving these instances exactly is mostly intractable, even 
with the algorithm of Dobkin et al.~\cite{DEM96}. 
Therefore, with the exception of the smallest instance,
it cannot be said if the results of this second series of
examples are good approximations of the real discrepancy of the GLP
sets under consideration or not.

\subsection{Comparisons against the Algorithm of Winker and Fang} \label{subsec:expwf}

\begin{table}
\footnotesize
  \centering
  \begin{tabular}{lcc }
              & Smaller instance & Larger instance \\
    Algorithm & $d=10$, $n=100$  & $d=20$, $n=1000$ \\
\hline
\basic\          & $0.78$s          & $9.34$s \\
\improved, $\delta$ only & $1.22$s & $10.94$s  \\
\improved, $\bar\delta$ only & $0.85$s & $9.11$s \\
\improved, mixed form & $1.87$s & $20.37$s \\
Winker \& Fang & $0.61$s & $7.2$s \\
\hline
  \end{tabular}
  \caption{\footnotesize Running times for the considered algorithms. All algorithms executed one trial of 100,000 iterations.  The inputs are two randomly generated point sets.}
  \label{tab:time}
\end{table}

To begin the comparisons, an indication of the running times of the algorithms
is given in Table~\ref{tab:time}.  
As can be seen from the table, \basic\ takes slightly more time than
our implementation of Winker and Fang, and \improved\ takes between
two and three times as long, mainly due to the snapping procedures.
For \improved, we report the separate times for~$\delta$ and~$\bar \delta$
optimization, as well as the time required for a mixed optimization of both
(as is done in \basic).  As can be seen, the overhead due to splitting is
negligible to non-existent. 

\begin{table}
  \centering
\footnotesize
  \begin{tabular}{lrrlrrrrrr}
         &   &    &                & \multicolumn{2}{c}{\basic\ } & \multicolumn{2}{c}{\improved\ } & \multicolumn{2}{c}{Winker \& Fang} \\
    Class&$d$&$n$ & $\disc(\cdot)$ & Hits & Best-of-10 & Hits & Best-of-10 &    Hits & Best-of-10 \\
\hline
4.145 & 4 & 145 & $0.0731$ &   99      & $0.0731$ &   100     & $0.0731$ & 7 & 0.0729 \\   4.255 & 4 & 255 & $0.1093$ &    98     & $0.1093$ &   100     & $0.1093$ & 35& 0.1093 \\   4.312 & 4 & 312 & $0.0617$ &   100     & $0.0617$ &   100     & $0.0617$ & 19& 0.0616 \\   4.376 & 4 & 376 & $0.0753$ &   30      & $0.0753$ &    79     & $0.0753$ & 0 & 0.0742 \\   4.388 & 4 & 388 & $0.1297$ &    58     & $0.1297$ &   100     & $0.1297$ & 0 & 0.1284 \\   4.443 & 4 & 443 & $0.0242$ &   38      & $0.0242$ &    90     & $0.0242$ & 0 & 0.0224 \\   4.448 & 4 & 448 & $0.0548$ &   47      & $0.0548$ &   100     & $0.0546$ & 0 & 0.0538 \\   4.451 & 4 & 451 & $0.0270$ &   0       & $0.0265$ &     8     & $0.0270$ & 0 & 0.0252 \\   4.471 & 4 & 471 & $0.0286$ &   39      & $0.0286$ &    99     & $0.0286$ & 0 & 0.0276 \\   4.487 & 4 & 487 & $0.0413$ &   24      & $0.0413$ &    93     & $0.0413$ & 0 & 0.0396 \\   \hline                                                                    
5.102 & 5 & 102 & $0.1216$ &    100    & $0.1216$ &       100 & $0.1216$ & 2 & 0.1193 \\ 5.122 & 5 & 122 & $0.0860$ &   8       & $0.0854$ &        58 & $0.0860$ & 0 & 0.0826 \\ 5.147 & 5 & 147 & $0.1456$ &    100    & $0.1456$ &       100 & $0.1456$ & 0 & 0.1418 \\ 5.153 & 5 & 153 & $0.1075$ &    100    & $0.1075$ &       100 & $0.1075$ & 1 & 0.1041 \\ 5.169 & 5 & 169 & $0.0755$ &   15      & $0.0752$ &        98 & $0.0755$ & 0 & 0.0691 \\ 5.170 & 5 & 170 & $0.0860$ &   81      & $0.0860$ &       100 & $0.0860$ & 0 & 0.0789 \\ 5.195 & 5 & 195 & $0.1574$ &     100   & $0.1574$ &       100 & $0.1574$ & 0 & 0.1533 \\ 5.203 & 5 & 203 & $0.1675$ &    100    & $0.1675$ &       100 & $0.1675$ & 0 & 0.1639 \\ 5.235 & 5 & 235 & $0.0786$ &   88      & $0.0786$ &       100 & $0.0786$ & 0 & 0.0706 \\ 5.236 & 5 & 236 & $0.0582$ &   7       & $0.0578$ &       74  & $0.0582$ & 0 & 0.0541 \\ \hline
6.28  & 6 &  28 & $0.5360$ &   100     & $0.5360$ &        33 & $0.5358$ & 100 & 0.5360\\  6.29  & 6 &  29 & $0.2532$ &   100     & $0.2532$ &       100 & $0.2532$ & 12&   0.2527\\  6.35  & 6 &  35 & $0.3431$ &   0       & $0.2859$ &        96 & $0.3431$ & 54 &  0.3431\\  6.50  & 6 &  50 & $0.3148$ &   4       & $0.3118$ &       100 & $0.3148$ & 59 &  0.3148\\  6.61  & 6 &  61 & $0.1937$ &   84      & $0.1937$ &       100 & $0.1937$ & 1 &   0.1872\\  6.73  & 6 &  73 & $0.1485$ &   28      & $0.1485$ &        95 & $0.1485$ & 0 &   0.1391\\  6.81  & 6 &  81 & $0.25$   &  24       & $0.2500$ &       100 & $0.25  $ & 0 &   0.2440\\  6.88  & 6 &  88 & $0.2658$ &   100     & $0.2658$ &       100 & $0.2658$ & 3 &   0.2608\\  6.90  & 6 &  90 & $0.1992$ &   100     & $0.1992$ &       100 & $0.1992$ & 23&   0.1990\\  6.92  & 6 &  92 & $0.1635$ &   100     & $0.1635$ &       100 & $0.1635$ & 5 &   0.1630\\  \hline
6.2129&  6 &2129& $0.0254$ &  0  & $0.0241$ &   13  & $0.0254$ & 0 & $0.0217$ \\  7.3997&7&3997& $0.0254^*$   & 0& $0.0222$ &     15  & $0.0254$ & 0 & $0.0218$ \\   8.3997&8&3997& $0.0254^*$   &0&  $0.0235$ &     10  & $0.0254$ & 0 & $0.0217$ \\  9.3997&9&3997& $0.0387^*$   &0&  $0.0366$ &      0  & $0.0375$ & 0 & $0.0354$ \\  10.4661&10&4661&$0.0272^*$  &0&  $0.0264$ &     40 &  $0.0272$ & 0 & $0.0230$ \\  11.4661&11&4661&$0.0283^*$  &0&  $0.0275$ &      3  & $0.0280$ & 0 & $0.0235$ \\  \hline
 \end{tabular}
  \caption{\footnotesize Data for GLP sets used by Winker and Fang~\cite{WF97}.
Discrepancy values marked with a star are lower bounds only (i.e., largest
discrepancy found over all executions of algorithm variants). 
All data is computed using $100$ trials of $100,000$ iterations;
reported is the average value of best-of-10 calls, and number of times (out of
100) that the optimum (or a value matching the largest known value) was found.
The data for Winker and Fang is for our re-implementation of the algorithm;
the original results for the same instances can be found in~\cite{WF97}.}
  \label{tab:glpwf}
\end{table}

The parameter settings for our implementation of the algorithm of Winker and Fang are as follows.
Since our experiments did not reveal a strong influence of the choice of $\alpha$ on the quality of the algorithm, we fix $\alpha:=0.995$ for our experiments.
Winker and Fang do not explicitly give a rule how one should choose $k$ and $mc$. 
For the small-dimensional data (Table~\ref{tab:glpwf}), we use the settings of~\cite{WF97}. 
For the other tests, we use~$mc=3$ if~$d\leq 12$ and~$mc=4$ otherwise, and~$k=41$ if~$n\leq 500$ and~$k=301$ otherwise.
This seems to be in line with the choices of Winker and Fang for the sizes used.

Table~\ref{tab:glpwf} shows the data for the GLP sets used by Winker and Fang in~\cite{WF97}.
Although the last group of point sets are quite large, note that this data is mostly of modest dimension.
As can be seen, for these sizes, all algorithms behave reasonably, with both
of our algorithms generally outperforming our implementation of Winker and Fang, 
and with \improved\ showing much higher precision than \basic.

We note that our re-implementation of the Winker and Fang algorithm gives
notably worse results than what was reported in~\cite{WF97} for the same
instances.  For the larger instances (i.e., with thousands of points), 200,000 
iterations are used in~\cite{WF97} while we use 100,000 iterations throughout, 
but there is also a clear difference for the smaller settings.
Adjusting the parameters of our implementation to match those used
in~\cite{WF97} has not been found to compensate for this.  After significant
experimentation, the best hypothesis we can provide is that there might be a
difference in the behavior of the pseudo-random number generators used (in
particular, as~\cite{WF97} uses a random number library we do not have access
to).  Still, even compared to the results reported in~\cite{WF97}, our
algorithms, and \improved\ in particular, still fare well.

\begin{table}
  \centering
\footnotesize
  \begin{tabular}{rrrcrrrrll}
         &   &    & $\disc(\cdot)$  & \multicolumn{2}{c}{\basic\ } & \multicolumn{2}{c}{\improved\ } & \multicolumn{2}{c}{Winker \& Fang } \\
    Name &$d$&$n$ &           found& Hits & Best-of-10 & Hits & Best-of-10 & Hits & Best-of-10 \\
\hline
Sobol' & 7 &  256& $0.0883$ & 1 & $0.0804$ & 78 & $0.0883$ & 0  & $0.0819$ \\
Sobol' & 7 &  512& $0.0452$ & 1 & $0.0440$ & 17 & $0.0451$ & 0  & $0.0395$ \\
Sobol' & 8 &  128& $0.1202$ & 0 & $0.1198$ & 98 & $0.1202$ & 0  & $0.1102$ \\
Sobol' & 9 &  128& $0.1372$ & 8 & $0.1367$ &100 & $0.1372$ & 0  & $0.1254$ \\
Sobol' & 10&  128& $0.1787$ &36 & $0.1787$ &100 & $0.1787$ & 0  & $0.1606$ \\
Sobol' & 11&  128& $0.1811$ &14 & $0.1811$ & 97 & $0.1811$ & 0  & $0.1563$ \\
Sobol' & 12&  128&$0.1885$& 1 & $0.1873$ & 82 & $0.1885$ & 0  & $0.1689$ \\
Sobol' & 12&  256&$0.1110^*$& 2 & $0.1108$ & 41 & $0.1110$ & 0  & $0.0908$ \\
\hline
Faure & 7 & 343 & $0.1298$ & 21& $0.1297$ & 100& $0.1298$ & 0  & $0.1143$ \\
Faure & 8 & 121 & $0.1702$ & 99& $0.1702$ & 100& $0.1702$ & 0  & $0.1573$ \\
Faure & 9 & 121 & $0.2121$ & 98& $0.2121$ & 100& $0.2121$ & 0  & $0.1959$ \\
Faure & 10& 121 & $0.2574$ & 95& $0.2574$ & 100& $0.2574$ & 0  & $0.2356$ \\
Faure & 11& 121 & $0.3010$ &100& $0.3010$ & 100& $0.3010$ & 0  & $0.2632$ \\
Faure & 12& 169 &$0.2718$& 73& $0.2718$ & 100& $0.2718$ & 0  & $0.1708$ \\
\hline
GLP   & 6 & 343 & $0.0870$ & 1 & $0.0869$ & 36 & $0.0870$ & 0  & $0.0778$ \\
GLP   & 7 & 343 & $0.0888$ & 3 & $0.0883$ & 28 & $0.0888$ & 0  & $0.0791$ \\
GLP   & 8 & 113 & $0.1422$ & 6 & $0.1399$ & 95 & $0.1422$ & 0  & $0.1303$ \\
GLP   & 9 & 113 & $0.1641$ &98 & $0.1641$ &100 & $0.1641$ & 0  & $0.1490$ \\
GLP   & 10& 113 & $0.1871$ & 1 & $0.1862$ & 94 & $0.1871$ & 0  & $0.1744$ \\
\hline
Sobol' &20&128   &$0.2616^*$& 0 & $0.2576$ & 51 & $0.2616$ & 0  & $0.0497$ \\
Sobol' &20&256   &$0.1856^*$&13 & $0.1854$ & 49 & $0.1856$ & 0  & $0.0980$ \\
Sobol' &20&512 & $0.1336^*$ & 0 & $0.1080$ & 86 & $0.1336$ & 0  & $0.0635$ \\
Sobol' &20&1024 &$0.1349^*$ & 0 & $0.0951$ &  0 & $0.1330$ & 0  & $0.0560$ \\
Sobol' &20&2048& $0.0724^*$ & 0 & $0.0465$ &  0 & $0.0505$ & 0  & $0.0370$ \\
Faure &20&529 & $0.2615^*$ & 0 & $0.2587$ & 98 & $0.2615$ & 0  & $0.0275$ \\
Faure &20&1500& $0.0740^*$ & 0 & $0.0733$ & 14 & $0.0740$ & 0  & $0.0347$ \\
GLP   &20&149 & $0.2581^*$ & 1 & $0.2548$ & 65 & $0.2581$ & 0  & $0.0837$ \\
GLP   &20&227 & $0.1902^*$ & 0 & $0.1897$ &  1 & $0.1899$ & 0  & $0.0601$ \\
GLP   &20&457 & $0.1298^*$ & 0 & $0.1220$ &  3 & $0.1272$ & 0  & $0.0519$ \\
GLP   &20&911 & $0.1013^*$ & 0 & $0.0975$ &  8 & $0.1013$ & 0  & $0.0315$ \\
GLP   &20&1619& $0.0844^*$ & 0 & $0.0809$ &  2 & $0.0844$ & 0  & $0.0299$ \\
\hline
Sobol' &50&2000& $0.1030^*$ & 0 & $0.0952$ &  0 & $0.1024$ & 0  & $0.0005$ \\
Sobol' &50&4000& $0.0677^*$ & 0 & $0.0597$ &  0 & $0.0665$ & 0  & $0.00025$ \\
Faure &50&2000& $0.3112^*$ & 0 & $0.2868$ &100 & $0.3112$ & 0  & $0.0123$ \\
Faure &50&4000& $0.1979^*$ & 0 & $0.1912$ & 0  & $0.1978$ & 0  & $0.0059$ \\
GLP   &50&2000& $0.1465^*$ & 0 & $0.1317$ & 0  & $0.1450$ & 0  & $0.0005$ \\
GLP   &50&4000& $0.1205^*$ & 0 & $0.1053$ & 0  & $0.1201$ & 0  & $0.0003$ \\
\hline
  \end{tabular}
  \caption{\footnotesize New instance comparisons. 
Discrepancy values marked with a star are lower bounds only (i.e., largest
discrepancy found over all executions of algorithm variants). 
All data is computed using $100$ trials of $100,000$ iterations;
reported is the average value of best-of-10 calls, and number of times (out of
100) that the optimum (or a value matching the largest known value) was found.}
  \label{tab:newlarge}
\end{table}

Table~\ref{tab:newlarge} shows the new data, for larger-scale instances. 
A few new trends are noticeable, in particular for the higher-dimensional
data.  Here, the algorithm of Winker and Fang seems to deteriorate, 
and there is also a larger difference emerging between \basic\ and \improved,
in particular for the Sobol' sets.  However, as can be seen for the 2048-point,
20-dimensional Sobol' set, it does happen that the lower bound is quite
imprecise.  (The value of~$0.0724$ for this point set was discovered only a
handful of times over nearly 5000 trials of algorithm variants and settings.)

The highest-dimensional sets ($d=50$) illustrate the deterioration of Winker
and Fang with increasing dimension; for many of the settings, the largest
error this algorithm finds is exactly~$1/n$ (due to the zero-volume box
containing the origin with one point).

\subsection{Comparisons with the Algorithm by Shah} \label{subsec:expshah}

\begin{table}
  \centering
\footnotesize
  \begin{tabular}{lrrlrrrrll}
         &   &    &                & \multicolumn{2}{c}{\basic\ } & \multicolumn{2}{c}{\improved\ } & \multicolumn{2}{c}{Shah} \\
    Class&$d$&$n$ & $\disc(\cdot)$ & Hits & Best-of-10 & Hits & Best-of-10 & Hits & Best Found \\
\hline
Halton & 5 & 50 & $0.1886$ & 100 & $0.1886$ & 100 & $0.1886$ & 81 & $0.1886$ \\
Halton & 7 & 50 & $0.2678$ & 100   & $0.2678$ & 100 & $0.2678$ & 22 & $0.2678$ \\
Halton & 7 & 100& $0.1714$ & 9 & $0.1710$ & 100 & $0.1714$ & 13 & $0.1714$ \\
Halton & 7 &1000& $0.0430$ & 0 & $0.0424$& 81 & $0.0430$ & $\phantom{0}8^{(1)}$ & $0.0430^{(1)}$ \\
Faure & 10 & 50 & $0.4680$ & 100 & $0.4680$ & 100 & $0.4680$ & 97 & $0.4680$ \\
Faure & 10 & 100& $0.2483$ & 52 & $0.2483$ & 100 & $0.2483$ &  28 & $0.2483$ \\ 
Faure & 10 & 500& $0.0717^*$& 2 & $0.0701$ &100& $0.0717$ & $\phantom{0}0^{(1)}$ & $0.0689^{(1)}$ \\
\hline
  \end{tabular}
  \caption{\footnotesize Comparison against point sets used by Shah. Reporting average value of best-of-10 calls, and number of times (out of 100) that the optimum was found; for Shah, reporting highest value found, and number of times (out of 100) this value was produced.
The discrepancy value marked with a star is lower bound only (i.e., largest value found by any algorithm). Values marked (1) are recomputed using the same settings as in~\cite{Sha10}.}
  \label{tab:shah}
\end{table}

Table~\ref{tab:shah} lists the point sets used by Shah~\cite{Sha10}.  
The Faure sets here are somewhat nonstandard in that they exclude the origin,
i.e., they consist of points 2 through~$n+1$ of the Faure sequence,
where the order of the points is as produced by the program of Burkhardt~\cite{faureurl}.

Some
very small point sets were omitted, as  
every reported algorithm would find the optimum every time.
For all but one of the point sets, the exact discrepancy could be computed;
the remaining instance is the first 500 points of the 10-dimensional Faure sequence.

Most of the sets seem too easy to really test the algorithms, i.e., all
variants frequently find essentially optimal points.  The one exception is the
last item, which shows a clear advantage for our algorithms.  We also find
(again) that \improved\ has a better precision than the other algorithms. 

\subsection{Comparisons with the Algorithms by Thi\'emard} \label{subsec:expthie}

\begin{table}
  \centering
\footnotesize
  \begin{tabular}{ccccccc}
\hline
         &  \multicolumn{2}{c}{\improved\ }& \multicolumn{2}{c}{Thi\'emard: Initial} & Same time,& Same result,\\
Instance & Time & Result                   & Time & Result                         &   result    & time \\
\hline
Faure-12-169& 25s & $0.2718$               & 1s   & $0.2718$                       & $0.2718$  & 1s \\
Sobol'-12-128& 20s & $0.1885$               & 1s   & $0.1463$                       & $0.1463$  & 453s ($7.6$m) \\
Sobol'-12-256& 35s & $0.1110$               & 3s   & $0.0872$                       & $0.0873$  & $1.6$ days \\
Faure-20-1500&280s (4.7m) & $0.0740$       & 422s (7m) & $0.0732$                  & None      & $>4$ days \\
GLP-20-1619  & 310s (5.2m) & $0.0844$      & 564s (9.4m) & $0.0572$                       & None      & $>5$ days \\
Sobol'-50-4000& 2600s (42m) & $0.0665$      & 32751s (9h) & $0.0743$                & None      & 32751s (9h) \\
GLP-50-4000  & 2500s (42m) & $0.1201$      & 31046s (8.6h) & $0.0301$              & None      & $>5$ days \\
\hline \\
  \end{tabular}
  \caption{Comparison against the integer programming-based algorithm of Thi\'emard~\cite{Thi01b}.
The values for \improved\ represent the time and average result of a best-of-10 computation
with $100,000$ iterations per trial. 
The middle pair of columns give the time required for~\cite{Thi01b} to return a first output,
and the value of this output; the last two columns report the lower bound reached by~\cite{Thi01b}
if allocated the same time that \improved\ needs for completion, and the time required by~\cite{Thi01b}
to match the result of \improved.}
  \label{tab:cmpthiemard}
\end{table}

Finally, we give a quick comparison against the integer programming-based
algorithm of Thi\'emard~\cite{Thi01b}.  Since~\cite{Thi01b} has the feature
that running it for a longer time produces gradually stronger bounds, we
report three different checkpoint values; see Table~\ref{tab:cmpthiemard} for details.
The results are somewhat irregular; however,~\cite{Thi01b} may require a lot
of time to report a first value, and frequently will not improve significantly
on this initial lower bound except after very large amounts of computation
time (for example, for the 12-dimensional, 256-point Sobol' set, the value~$0.0872$ 
is discovered in seconds, while the first real improvement takes over an hour
to produce).

Thi\'emard also constructed a second algorithm for discrepancy estimation,
based on delta-covers~\cite{Thi01a}; this is freely downloadable from Thi\'emard's homepage.
Its prime feature is that it provides upper bounds with a non-trivial running time guarantee.
The lower bounds that it produces are not as helpful as the upper bounds, e.g.,
it was reported in~\cite{DGW10} and~\cite{Sha10} that the lower
bounds from the preliminary version of \basic~\cite{Win07} and the genetic
algorithm of Shah~\cite{Sha10} are better. Thus we omit this kind of comparison here.

\section{Conclusion}\label{sec:outlook}Our numerical experiments clearly indicate that the improvements made from the algorithm of 
Winker and Fang in \basic\ and \improved\ greatly increases the quality of the lower bounds,
in particular for the difficult higher-dimensional problem instances.
Nevertheless, we do not fully understand the behavior of different algorithm variants with regards to \emph{snapping}.
In particular, one might well have expected the variant described in Section~\ref{sssec:variant} to do better than the one of Section~\ref{subsec:snapping}
that we currently use.  It is still possible that a judicious application of the ``snap-move'' variant of Section~\ref{sssec:variant},
perhaps only in certain situations, can improve the behavior further.

Still, all in all, we conclude that the algorithms \basic\ and \improved\ presented in the current work represent significant improvements
over previous lower-bound heuristics for computing the star discrepancy, and to the best of our knowledge, make up the best performing
star discrepancy estimation algorithms available.

\section*{Acknowledgments}
We gratefully acknowledge Manan Shah, Eric Thi\'emard, and Peter Winker for providing us 
with source code and implementations of the applications used in their experiments 
(in~\cite{Sha10},~\cite{Thi01b}, and~\cite{WF97}, respectively), 
and in general for helpful comments. 

Michael Gnewuch was supported by the German Research Foundation (DFG) under grants GN 91/3-1 and GN 91/4-1. Part of his work was done while he was at the Max Planck Institute for Mathematics in the Sciences in Leipzig and at Columbia University in the City of New York.

Magnus Wahlstr\"om is supported by the DFG via its priority program "SPP 1307: Algorithm Engineering'' under grant DO 749/4-1.

Carola Winzen is a recipient of the Google Europe Fellowship in Randomized Algorithms, and this research is supported in part by this Google Fellowship.

\bibliographystyle{amsalpha}

\providecommand{\bysame}{\leavevmode\hbox to3em{\hrulefill}\thinspace}
\providecommand{\MR}{\relax\ifhmode\unskip\space\fi MR }
\providecommand{\MRhref}[2]{  \href{http://www.ams.org/mathscinet-getitem?mr=#1}{#2}
}
\providecommand{\href}[2]{#2}

\newpage
\appendix
\label{SEC:APPENDIX}

\newpage
\section{Calculation of $\lambda^d(\A(z))$ and $\pi^d(\A(z))$}
\label{APPE}

\begin{lemma}
\label{lemmaLB1}
Let $\varepsilon \in (0,1]$, and let $z\in[0,1]^d$ with $V_z
\geq \varepsilon$. Then
\begin{equation}
\label{voila}
\lambda^d(\A(z)) = V_z - (V_z - \varepsilon) \sum^{d-1}_{k=0}
\frac{(-\ln(1-\varepsilon/V_z))^k}{k!} \,.
\end{equation}
\end{lemma}

\begin{proof}
Let $V_z\geq \varepsilon$. Then we have
\begin{equation*}
\lambda^d(\A(z)) = \int^{z_1}_{\alpha_1}...\int^{z_d}_{\alpha_d}\,
d\zeta_d...\,d\zeta_1\,,
\end{equation*}
where 
\begin{equation*}
\alpha_1 = \frac{V_z -
\varepsilon}{z_{2} z_3...z_d}\,,\;
\alpha_2 = \frac{V_z -
\varepsilon}{\zeta_1 z_{3}...z_d}\, ,...,
\;\alpha_d = \frac{V_z -
\varepsilon}{\zeta_1 \zeta_2 ...\zeta_{d-1}}\,.
\end{equation*}
We prove formula (\ref{voila}) by induction over the dimension $d$. If $d=1$, then 
clearly $\lambda(A_\varepsilon(z)) = \varepsilon$.
Let now $d\geq 2$. We denote by $\tilde{z}$ the $(d-1)$-dimensional
vector $(z_2,...,z_d)$ and by $\tilde{\varepsilon}$ the term
$(\varepsilon + (\zeta_1 - z_1)V_{\tilde{z}})/\zeta_1$.
Furthermore we define for $i\in [d-1]$ the lower integration limit
$\tilde{\alpha}_i = (V_{\tilde{z}} -
\tilde{\varepsilon})/(\zeta_2...\zeta_{i}\,\tilde{z}_{i+1}...\tilde{z}_{d-1})$.
Note that $\tilde{\alpha}_i = \alpha_{i+1}$.
Then, by our induction hypothesis,
\begin{equation*}
\begin{split}
\lambda^d(\A(z)) &= \int^{z_1}_{\alpha_1}
\int^{\tilde{z}_1}_{\tilde{\alpha}_1}... 
\int^{\tilde{z}_{d-1}}_{\tilde{\alpha}_{d-1}}\,d\zeta_{d}...\,
d\zeta_{2}\,d\zeta_1\\
&= \int^{z_1}_{\alpha_1} \Bigg( V_{\tilde{z}} - (V_{\tilde{z}} -
\tilde{\varepsilon})\sum^{d-2}_{k=0} 
\frac{(-\ln(1 - \tilde{\varepsilon}/V_{\tilde{z}}))^k}{k!} \Bigg)
\,d\zeta_1\\
&= V_z - (V_z - \varepsilon) - (V_z - \varepsilon) \Bigg[ \sum^{d-1}_{k=1}
\frac{1}{k!} \ln \Big( \frac{V_{\tilde{z}}}{V_z - \varepsilon}\zeta_1 
\Big)^k \Bigg]^{z_1}_{\zeta_1 = \alpha_1}\\
&=  V_z - (V_z - \varepsilon) \sum^{d-1}_{k=0}
\frac{(-\ln(1-\varepsilon/V_z))^k}{k!} \,.
\end{split}
\end{equation*}
\end{proof}

\begin{proposition}
\label{lemmaLB3}
Let $d\geq 2$. For $z\in [0,1]^d$ with 
$V_z > \varepsilon$, we obtain
\begin{equation*}
\lambda^d(\A(z)) = \frac{1}{d!}\frac{\varepsilon^d}{V^{d-1}_z}
\sum^\infty_{k=0} b_k(d) \Big( \frac{\varepsilon}{V_z} \Big)^k
\end{equation*}
with positive coefficients
\begin{equation*}
b_k(2) = \frac{2}{(k+1)(k+2)}\,,\hspace{2ex}
b_k(3) = \frac{6}{(k+2)(k+3)} \sum^k_{\nu=0} \frac{1}{\nu+1}
\end{equation*}
and
\begin{equation*}
b_k(d) = \frac{d!}{(k+d-1)(k+d)}\sum^k_{k_1 =
0}...\sum^{k_{d-3}}_{k_{d-2}=0}\, \prod^{d-2}_{j=1}\,\frac{1}{k_j + d -
  j-1}\hspace{2ex}\text{for $d\geq 4$.}
\end{equation*}
The power series converges for each $\epsilon > 0$ uniformly and absolutely
for all $V_z \in [\varepsilon + \epsilon, 1]$.
Furthermore, we have $b_0(d) = 1$, $b_1(d) = d(d-1)/2(d+1)$, and for
all $k$ the inequality 
$b_k(d) \leq d^k/2^{k-1}$ is satisfied.
\end{proposition}

\begin{proof}
To prove the  power series expansion, we consider the function
\begin{equation*}
R(x,d) = 1 - (1-x)\sum^{d-1}_{k=0} \frac{(-\ln(1-x))^k}{k!}\hspace{2ex}
\text{for $x\in [0,1)$.}
\end{equation*}
Due to Lemma \ref{lemmaLB1} we have $\lambda^d(A_\varepsilon(z)) 
= V_z R(\varepsilon/V_z, d)$.
Since $\partial_x R(x,d) = (-\ln(1-x))^{d-1}/(d-1)!$, it suffices to
prove the following statement by induction over $d$:
\begin{equation}
\label{lnhochd}
\frac{(-\ln(1-x))^{d-1}}{(d-1)!} = \frac{1}{d!}\sum^\infty_{k=0}
(k+d) b_k(d)x^{k+d-1}\,,
\end{equation}
where the power series converges for each $\epsilon > 0$ uniformly and 
absolutely on $[0,1-\epsilon]$.
Let first $d=2$. Then
\begin{equation*}
-\ln(1-x) = \sum^\infty_{k=1}\frac{x^k}{k} 
= \frac{1}{2!}\sum^\infty_{k=0}(k+2)b_k(2)x^{k+1}\,,
\end{equation*}
and the required convergence of the power series is obviously given.
Now let $d\geq 3$. Our induction hypothesis yields
\begin{equation*}
\begin{split}
\partial_x \frac{(-\ln(1-x))^{d-1}}{(d-1)!} &= \frac{1}{1-x}
\frac{(-\ln(1-x))^{d-2}}{(d-2)!}\\
&= \bigg( \sum^\infty_{\nu = 0} x^\nu \bigg) \bigg(\frac{1}{(d-1)!}
\sum^\infty_{\mu = 0}(\mu + d -1) b_\mu(d-1)x^{\mu+d-2} \bigg)\\
&= \frac{1}{d!}\sum^\infty_{k=0} \bigg( d\sum^k_{\mu = 0} 
(\mu+d-1)b_\mu(d-1) \bigg) x^{k+d-2}\,,
\end{split}
\end{equation*}
where the last power series converges as claimed above.
Now
\begin{equation*}
\begin{split}
d\sum^k_{\mu=0} (\mu+d-1)b_\mu(d-1) &= \sum^k_{\mu=0} 
\frac{d!}{(\mu+d-2)}\sum^\mu_{\mu_1=0}...\sum^{\mu_{d-4}}_{\mu_{d-3}=0}
\prod^{d-3}_{j=2} \frac{1}{\mu_j+d-2-j}\\
&= d! \sum^k_{\nu_1 = 0}\sum^{\nu_1}_{\nu_2 = 0}...
\sum^{\nu_{d-3}}_{\nu_{d-2} = 0}\prod^{d-2}_{j=1} \frac{1}{\nu_j+d-j-1}\\
&= (k+d)(k+d-1)b_k(d)\,.
\end{split}
\end{equation*}
After integration we get (\ref{lnhochd}).

Furthermore, it is easily seen that $b_0(d) = 1$ and $b_1(d) =
d(d-1)/2(d+1)$. To complete the proof, we verify $b_k(d) \leq d^k/2^{k-1}$
for $k\geq 2$. 
The inequality is obviously true
in dimension $d=2$. Hence let $d\geq 3$. From the identity
\begin{equation*}
\sum^k_{\nu_1 = 0}...\sum^{\nu_{d-3}}_{\nu_{d-2} = 0} 1 = {k+d-2
  \choose k}
\end{equation*}
we obtain
\begin{equation*}
\sum^k_{\nu_1 = 0}\sum^{\nu_1}_{\nu_2 = 0}...
\sum^{\nu_{d-3}}_{\nu_{d-2} = 0}\prod^{d-2}_{j=1} \frac{1}{\nu_j+d-1-j}
\leq \frac{1}{(d-2)!}{k+d-2 \choose k}\,,
\end{equation*}
which leads to
\begin{equation*}
\begin{split}
b_k(d) &\leq
\frac{d(d-1)}{(k+d)(k+d-1)}\frac{(k+d-2)...(1+d-2)}{k...1}\\
&\leq
\frac{d(d-1)}{(k+d)(k+d-1)} \Big( \frac{d}{2} \Big)^{k-1} (d-1) \leq
\frac{d^k}{2^{k-1}}\,.
\end{split}
\end{equation*}
\end{proof}

\begin{corollary}
\label{Simplex}
Let $z\in [0,1]^d$.
If $V_z \geq d\varepsilon$, then
\begin{equation}
\label{simplex}
\lambda^d(\A(z)) \leq
\frac{5}{2\,d!}\frac{\varepsilon^d}{\,V^{d-1}_z}\,.
\end{equation}
\end{corollary}

We now consider the polynomial product measure $\pi^d$.

\begin{lemma}
\label{B1}
Let $\varepsilon \in (0,1]$, and let $z\in[0,1]^d$ with $V_z
\geq \varepsilon$. Then
\begin{equation}
\label{b1}
\pi^d(\A(z)) = V^d_z - (V_z - \varepsilon)^d \sum^{d-1}_{k=0}
\frac{d^k}{k!}(-\ln(1-\varepsilon/V_z))^k  \,,
\end{equation}
and, as a function of $V_z$, $\pi^d(\A(z))$ is strictly increasing. 
\end{lemma}

\begin{proof}
Let $V_z \geq \varepsilon$. 
We have
\begin{equation}
\label{transfor}
\pi^d(\A(z)) = \int_{\A(z)} d^d V^{d-1}_x \,\lambda^d(dx) 
= \int^\varepsilon_0 G(V_z, r)\,dr\,,
\end{equation}
where
\begin{equation*}
G(V_z,r) := d^d (V_z-r)^{d-1}\partial_r \lambda^d(A_r(z))\,.
\end{equation*}
From (\ref{voila}) we get for all 
$0\leq r \leq \varepsilon$
\begin{equation*}
\partial_r \lambda^d(A_r(z)) = \frac{ \big( -\ln(1-r/V_z) \big)^{d-1}}
{(d-1)!}\,.
\end{equation*}
If we define 
\begin{equation*}
F(r) := -(V_z - r)^d \sum^{d-1}_{k=0} \frac{d^k}{k!} 
(-\ln(1-r/V_z))^k\,,
\end{equation*}
then we observe that $F'(r) = G(V_z, r)$ holds. 
Thus we have $\pi^d(\A(z)) = F(\varepsilon) - F(0)$, which proves
(\ref{b1}). 
Furthermore, according to (\ref{transfor}), we get
\begin{equation*}
\partial_{V_z} \pi^d(\A(z)) = \int^\varepsilon_0 
\partial_{V_z} G(V_z,r)\,dr\,.
\end{equation*}
The integrand of the integral is positive, as the next calculation
reveals:
\begin{equation*}
\begin{split}
\partial_{V_z} G(V_z,r) &=
d^d (V_z - r)^{d-2} \frac{(-\ln(1-r/V_z))^{d-2}}{(d-2)!}
\big( -\ln(1-r/V_z) - r/V_z \big)\\
&= d^d (V_z - r)^{d-2} \frac{(-\ln(1-r/V_z))^{d-2}}{(d-2)!}
\sum^\infty_{k=2} \frac{1}{k} \Big( \frac{r}{V_z} \Big)^k > 0
\end{split}
\end{equation*}
for all $0 < r < V_z$. Thus $\partial_{V_z} \pi^d(\A(z)) > 0$, and, considered
as a function of $V_z$, $\pi^d(\A(z))$ is strictly increasing.
\end{proof}

\begin{proposition}
\label{B2}
Let $\varepsilon \in (0,1]$, and let $z\in [0,1]^d$ with $V_z \geq \varepsilon$.
Then we have the lower bound $\pi^d(\A(z)) \geq \varepsilon^d$. 
If furthermore $V_z \geq d\varepsilon$, then we have the estimate
\begin{equation*}
e^{-1} \frac{d^d}{d!} \varepsilon^d \leq \pi^d(\A(z)) \leq \frac{5}{2} 
\frac{d^d}{d!} \varepsilon^d\,.
\end{equation*}
\end{proposition}

\begin{proof}
Let $V_z = \varepsilon$. Then
\begin{equation*}
\pi^d(\A(z)) = \int_{[0, z]} d^d V^{d-1}_x 
\,\lambda^d(dx) = \prod^d_{i=1} z^d_i = \varepsilon^d\,.
\end{equation*}
Since $\pi^d(\A(z))$ is an increasing function of $V_z$, the first lower bound holds.
Let now $V_z \geq d\varepsilon$. Here we use the simple estimate
\begin{equation*}
d^d (V_z-\varepsilon)^{d-1} \lambda^d(\A(z)) \leq \pi^d(\A(z))
\leq d^d V_z^{d-1} \lambda^d(\A(z))\,.
\end{equation*}
Together with Proposition \ref{lemmaLB3} and Corollary \ref{Simplex} this leads 
to 
\begin{equation*}
e^{-1} \frac{d^d}{d!} \varepsilon^d \leq (1-1/d)^{d-1} \frac{d^d}{d!}\varepsilon^d 
\leq \pi^d(\A(z)) \leq  \frac{5}{2}\frac{d^d}{d!}\varepsilon^d\,.   
\end{equation*}
\end{proof}

\begin{remark}
Like $\lambda^d(\A(z))$ in Proposition \ref{lemmaLB3}, one can also
expand $\pi^d(\A(z))$ into a power series. This leads to 
\begin{equation*}
\pi^d(\A(z)) = \frac{d^d}{d!} \varepsilon^d \sum^\infty_{k=0} 
a_k(d) \Big( \frac{\varepsilon}{V_z} \Big)^k\,,
\end{equation*}
but here the coefficients $a_k(d)$ are not all positive. So we have,
e.g., $a_0(d) = 1$, but $a_1(d) = -d(d-1)/2(d+1)$. Therefore the 
power series expansion is here less useful than in the situation
of Proposition \ref{lemmaLB3}.
\end{remark}
}\end{document}